\pdfoutput=1
\documentclass[a4paper,UKenglish,cleveref, autoref,numberwithinsect]{lipics-v2019}

\title{Fine-Grained Complexity of Regular~Expression Pattern~Matching and Membership}

\author{Philipp Schepper}{
CISPA Helmholtz Center for Information Security, Saarbrücken, Germany
\and
Saarbrücken Graduate School of Computer Science, Saarland Informatics Campus, Saarbrücken, Germany}{philipp.schepper@cispa.saarland}{}{}

\authorrunning{P.~Schepper}

\Copyright{Philipp Schepper}



\relatedversion{
Full version of the paper accepted at ESA 2020 \cite{Schepper20}.
All presented lower bounds and an alternative proof of the upper bounds for pattern matching using the polynomial method are contained in the author's Master's thesis.}

\funding{Supported by the European Research Council (ERC) consolidator grant No.~725978 SYSTEMATICGRAPH.}

\acknowledgements{I thank Karl Bringmann for the supervision during the research for my Master's Thesis which this paper is based on and especially the pointer to \BOV which simplified the upper bounds extremely.}

\nolinenumbers

\hideLIPIcs

\usepackage{stmaryrd}

\usepackage{mathtools}
\usepackage{xspace}

\usepackage{tikz}
\usetikzlibrary{ipe}

\usepackage{xcolor}
\definecolor{ownDarkblue}{cmyk}{1,0.5,0.0,0.72}
\definecolor{ownRed}{cmyk}{0.29,0.93,0.96,0.05}
\colorlet{own}{ownRed}
\colorlet{snd}{ownDarkblue!60}

\theoremstyle{plain}
\newtheorem{hypothesis}[theorem]{Hypothesis}

\DeclarePairedDelimiter{\bin}{\langle}{\rangle}

\let\abs\relax
\DeclarePairedDelimiter\abs{\lvert}{\rvert}

\let\originalleft\left
\let\originalright\right
\def\left#1{\mathopen{}\originalleft#1}
\def\right#1{\originalright#1\mathclose{}}

\newcommand{\FormSat}{\textsc{Formula-SAT}\xspace}
\newcommand{\FormPair}{\textsc{Formula-Pair}\xspace}
\newcommand{\SETH}{\textsc{Strong~Exponential~Time~Hypothesis}\xspace}
\newcommand{\FSH}{\textsc{Formula-SAT~Hypothesis}\xspace}
\newcommand{\FPH}{\textsc{Formula-Pair~Hypothesis}\xspace}
\newcommand{\OV}{\textsc{Orthogonal~Vectors}\xspace}
\newcommand{\BOV}{\textsc{Batch-OV}\xspace}

\newcommand{\Conc}{\ensuremath{\circ}}
\newcommand{\Or}{\ensuremath{\mid}}
\newcommand{\Star}{\ensuremath{\star}}
\newcommand{\Plus}{\ensuremath{+}}

\newcommand{\newextmathcommand}[2]{%
  \newcommand{#1}{\xspace\ensuremath{#2}\xspace}
}

\newcommand{\sqlog}[1]{\sqrt{\log{#1}}}

\newcommand{\deff}{\coloneqq}

\newcommand{\starred}[1]{\overset{\star}{#1}}
\newcommand{\QED}{\lipicsEnd}

\DeclareMathOperator{\poly}{poly}

\newcommand{\prob}[1]{\Pr\left[ #1 \right]}

\newextmathcommand{\SetB}{\{0,1\}}

\newcommand{\lang}{\mathcal L}
\newcommand{\match}{\mathcal M}

\renewcommand{\O}{\mathcal O}
\newcommand{\true}{\textsf{true}\xspace}
\def\ee{\mathrm{e}}

\begin{document}
\maketitle
\begin{abstract}
  The currently fastest algorithm for regular expression pattern matching and membership improves the classical $\O(nm)$ time algorithm by a factor of about $\log^{3/2}n$.
  Instead of focussing on general patterns we analyse homogeneous patterns of bounded depth in this work.
  For them a classification splitting the types in easy (strongly sub-quadratic) and hard (essentially quadratic time under SETH) is known.
  We take a very fine-grained look at the hard pattern types from this classification and show a dichotomy:
  few types allow super-poly-logarithmic improvements
  while the algorithms for the other pattern types can only be improved by a constant number of log-factors, assuming the \FSH.
\end{abstract}

\section{Introduction}\label{sec:intro}
Regular expressions with the operations alternative \Or, concatenation \Conc, Kleene Plus \Plus, and Kleene Star \Star\
are used in many fields of computer science.
For example to search in texts and files or to replace strings by other strings as the unix tool \texttt{sed} does.
But they are also used to analyse XML files \cite{DBLP:conf/vldb/LiM01,DBLP:conf/pods/Murata01},
for network analysis \cite{DBLP:conf/icde/JohnsonMR07,DBLP:conf/ancs/YuCDLK06},
human computer interaction \cite{DBLP:conf/chi/KinHDA12},
and in biology to search for proteins in DNA sequences \cite{10.1093/nar/20.11.2861,doi:10.1089/106652703322756140}.

The most intuitive problem for regular expressions is the \emph{membership} problem.
There we ask whether a given text $t$ can be generated by a given regular expression $p$, i.e.\ is $t\in \lang(p)$?
We also call $p$ a pattern in the following.
A similar problem is the \emph{pattern matching} problem, where we are interested whether some \emph{substring} of the given text $t$ can be matched by $p$.
To simplify notation we define the matching language of $p$ as $\match(p)\deff \Sigma^* \lang(p)\Sigma^*$.
Then we want to check whether $t \in \match(p)$.
The standard algorithm for both problems runs in time $\O(nm)$ where $n$ is the text length and $m$ the pattern size \cite{Thompson68}.

Based on the ``Four Russians'' trick Myers showed an algorithm with running time $\O(nm/\log n)$ \cite{Myers92}.
This result was improved to an $\O(nm \log\log n/\log^{3/2}n)$ time algorithm by Bille and Thorup \cite{BilleT09}.
Although for several special cases of pattern matching and membership improved sub-quadratic time algorithms have been given \cite{AhoC75,ColeH02,KnuthMP77},
it remained an open question whether there are truly sub-quadratic time algorithms for the general case.
The first conditional lower bounds were shown by Backurs and Indyk \cite{BackursI16}.
They introduced so-called homogenous patterns and classified their hardness into easy, i.e.\ strongly sub-quadratic time solvable,
and hard, requiring essentially quadratic time assuming the \SETH~(SETH).
This classification of Backurs and Indyk was completed by a dichotomy for all homogeneous pattern types
by Bringmann, Gr{\o}nlund, and Larsen \cite{BringmannGL17}.
They reduced the hardness of all hard pattern types to the hardness of few pattern types of bounded depth.
By this it was sufficient to check few cases instead of infinitely many.

To understand what a homogeneous pattern is, we observe that one can see patterns as rooted and node labeled trees
where the inner nodes correspond to the operations of the pattern.
Then a pattern is homogenous if the operations on each level of the tree are equal.
The type of the pattern is the sequence of operations from the root to the leaves.
See \cref{sec:prelim} for a formal introduction.

But as SETH rules out only polynomial improvements, super-poly-logarithmic runtime improvements are still feasible.
Such improvements are know for \OV (OV) \cite{AbboudWY15,ChanW16}, for example, although there is a known conditional lower bound based on SETH.
But for pattern matching and membership no faster algorithms are known.
By a reduction from \FormSat Abboud and Bringmann showed
that in general pattern matching and membership cannot be solved in time $\O(nm/\log^{7+\epsilon} n)$
under the \FSH \cite{AbboudB18}.

For \FormSat one is given a De~Morgan formula $F$ over $n$ inputs and size $s$,
i.e.\ the formula is a tree where each inner gate computes the AND or OR of two other gates
and each of the $s$ leaves is labeled with one of the $n$ variables or their negation.
The task is to find a satisfying assignment for $F$.
While the naive approach takes time $\O(2^n s)$ to evaluate $F$ on all possible assignments,
there are polynomial improvements for formulas of size $s=o(n^3)$ \cite{ChenKS14,KomargodskiRT13,Santhanam10}.
But despite intense research there is currently no faster algorithm known for $s=n^{3+\Omega(1)}$.
Thus it seem reasonable to assume the following hypothesis:
\begin{hypothesis}[\FSH~(FSH) \cite{AbboudB18}]\label{hypo:FSH}
  There is no algorithm that can solve \FormSat on De~Morgan formulas of size $s = n^{3+\Omega(1)}$ in $\O(2^n/n^\epsilon)$ time,
  for some $\epsilon > 0$, in the Word-RAM model.
\end{hypothesis}
Although the new lower bound of $\O(nm/\log^{7+\epsilon} n)$ is quite astonishing since before only polynomial improvements have been ruled out,
the bound is for the general case.
It remained an open question whether it also holds for homogeneous patterns of bounded depth.
Using the results by Bringmann, Gr{\o}nlund, and Larsen \cite{BringmannGL17}
relating the hardness of different pattern types to each other,
it suffices to check the pattern types in \cref{tab:intro:hard} for the corresponding problem.

\begin{table}[t]
  \caption{Hard pattern types that have to be considered.}
  \label{tab:intro:hard}
  \centering
  \begin{tabular}{l||l|l|l|l|l|r|r|l}
      Pattern matching &
      \multirow{2}{*}{\Conc\Star} &
      \multirow{2}{*}{\Conc\Or\Conc} &
      \multirow{2}{*}{\Conc\Or\Plus} &
      \multirow{2}{*}{\Conc\Plus\Conc} &
      \multirow{2}{*}{\Conc\Plus\Or} &
      \Or\Conc\Or &
      \Or\Conc\Plus & \\
      \cline{1-1}
      \cline{7-9}
      Membership & & & & & &
      \Plus\Or\Conc\Or &
      \Plus\Or\Conc\Plus &
      \Or\Plus\Or\Conc
  \end{tabular}
\end{table}

We answer this last question and give a dichotomy for these hard pattern types:
For few pattern types we give the currently fastest algorithm for pattern matching and membership.
For the remaining patterns we show improved lower bounds of the form $\Omega(nm/\log^c n)$.
Where $c$ is a ``small'' constant only depending on the type of the pattern that arises from our reductions.

\begin{theorem}\label{thm:main}
  For texts of length $n$ and patterns of size $m$ we have the following time bounds for the  stated problems:
  \begin{itemize}
    \item
    ${nm} / {2^{\Omega(\sqlog{\min(n,m)})}}$ for
    \Or\Conc\Or- and \Or\Conc\Plus-pattern matching, and
    \Plus\Or\Conc\Or- and \Plus\Or\Conc\Plus-membership

    \item
    $\Theta({nm} / {\poly \log n})$ for
    pattern matching and membership with types
    \Conc\Plus\Or, \Conc\Or\Plus, \Conc\Plus\Conc, \Conc\Or\Conc, and \Conc\Star\
    and for \Or\Plus\Or\Conc-membership,
    unless FSH is false.
  \end{itemize}
\end{theorem}
This dichotomy result gives us a simple classification for the hard pattern types.
Depending on the pattern type one can decide if there is super-poly-logarithmic algorithm,
or if even the classical algorithm is optimal up to a constant number of log-factors.
See \cref{fig:intro:matching} for an overview of the results for pattern matching.
The corresponding figures for membership are shown in \cref{appendix:figures}.
\begin{figure}[t]
  \begin{tikzpicture}[ipe import]
    \node[ipe node, anchor=north west]
       at (324, 540) {
         \begin{minipage}{112bp}\kern0pt
           \color{own} $\frac{nm}{2^{\Omega(\sqlog{\min(n,m)})}}$ \\
           Thm.~\ref{thm:upper:main}
         \end{minipage}
       };
    \node[ipe node, anchor=north west]
       at (324, 636) {
         \begin{minipage}{112bp}\kern0pt
           \color{own} $\frac{nm}{2^{\Omega(\sqlog{\min(n,m)})}}$ \\
           Thm.~\ref{thm:upper:main}
         \end{minipage}
       };
    \node[ipe node, anchor=north west]
       at (324, 828) {
         \begin{minipage}{104bp}\kern0pt
           \color{own} $\Theta\left(\frac{nm}{\poly \log n}\right)$ \\
           Sec.~\ref{lower:concPlusConc}
         \end{minipage}
       };
    \node[ipe node, anchor=north west]
       at (324, 796) {
         \begin{minipage}{104bp}\kern0pt
           \color{snd} $\Theta\left(\frac{nm}{\poly \log n}\right)$ (\Conc\Star) \\
           Sec.~\ref{lower:concStar}, Lem.~\ref{lem:prelim:hardness}
         \end{minipage}
       };
    \node[ipe node, anchor=north west]
       at (324, 764) {
         \begin{minipage}{104bp}\kern0pt
           \color{own} $\Theta\left(\frac{nm}{\poly \log n}\right)$ \\
           Sec.~\ref{lower:concPlusOr}
         \end{minipage}
       };
    \draw
      (300, 784)
       -- (320, 784);
    \node[ipe node]
       at (312, 820) {\Conc};
    \node[ipe node]
       at (312, 788) {\Star};
    \node[ipe node]
       at (312, 756) {\Or};
    \node[ipe node, anchor=north west]
       at (224, 796) {
         \begin{minipage}{80bp}\kern0pt
           $\O(n \log^2 m + m)$ \\
           \cite{BackursI16}
         \end{minipage}
       };
    \node[ipe node, anchor=north west]
       at (324, 700) {
         \begin{minipage}{104bp}\kern0pt
           \color{own} $\Theta\left(\frac{nm}{\poly \log n}\right)$ \\
           Sec.~\ref{lower:concOrConc}
         \end{minipage}
       };
    \node[ipe node, anchor=north west]
       at (324, 668) {
         \begin{minipage}{104bp}\kern0pt
           \color{snd} $\Theta\left(\frac{nm}{\poly \log n}\right)$ (\Conc\Star) \\
           Sec.~\ref{lower:concStar}, Lem.~\ref{lem:prelim:hardness}
         \end{minipage}
       };
    \node[ipe node, anchor=north west]
       at (324, 732) {
         \begin{minipage}{104bp}\kern0pt
           \color{own} $\Theta\left(\frac{nm}{\poly \log n}\right)$ \\
           Sec.~\ref{lower:concOrPlus}
         \end{minipage}
       };
    \draw
      (300, 688)
       -- (320, 688);
    \node[ipe node]
       at (312, 692) {\Conc};
    \node[ipe node]
       at (312, 660) {\Star};
    \node[ipe node]
       at (312, 724) {\Plus};
    \node[ipe node, anchor=north west]
       at (224, 700) {
         \begin{minipage}{80bp}\kern0pt
           $\O(n \log^2 m + m)$ \\
           \cite{ColeH02}
         \end{minipage}
       };
    \node[ipe node, anchor=north west]
       at (224, 748) {
         \begin{minipage}{80bp}\kern0pt
           \color{own} $\Theta\left(\frac{nm}{\poly \log n}\right)$ \\
           Sec.~\ref{lower:concStar}
         \end{minipage}
       };
    \draw
      (200, 736)
       -- (220, 736);
    \node[ipe node, anchor=north west]
       at (152, 756) {
         \begin{minipage}{80bp}\kern0pt
           String\\Matching \\
           $\Theta(n+m)$ \\
           \cite{KnuthMP77}
         \end{minipage}
       };
    \node[ipe node]
       at (140, 740) {\Conc};
    \node[ipe node]
       at (212, 740) {\Star};
    \node[ipe node]
       at (212, 788) {\Plus};
    \node[ipe node]
       at (212, 692) {\Or};
    \draw
      (148, 736)
       -- (136, 736);
    \node[ipe node, anchor=north west]
       at (152, 672) {
         \begin{minipage}{100bp}\kern0pt
           Complete Subtree \\
           $\Theta(n+m)$ \\
           immediate
         \end{minipage}
       };
    \node[ipe node, anchor=north west]
       at (152, 824) {
         \begin{minipage}{80bp}\kern0pt
           Simplifies \\
           Lem.~\ref{lem:prelim:simplification}
         \end{minipage}
       };
    \node[ipe node, anchor=north west]
       at (324, 588) {
         \begin{minipage}{104bp}\kern0pt
           \color{snd} $\Theta\left(\frac{nm}{\poly \log n}\right)$ (\Conc\Star) \\
           Sec.~\ref{lower:concStar}, Lem.~\ref{lem:prelim:hardness}
         \end{minipage}
       };
    \draw
      (272, 576)
       -- (320, 576);
    \node[ipe node]
       at (312, 580) {\Star};
    \node[ipe node]
       at (312, 628) {\Plus};
    \node[ipe node]
       at (312, 532) {\Or};
    \draw
      (220, 576)
       -- (200, 576);
    \node[ipe node, anchor=north west]
       at (224, 596) {
         \begin{minipage}{96bp}\kern0pt
           Dictionary\\Matching \\
           $\Theta(n+m)$ \\
           \cite{AhoC75}
         \end{minipage}
       };
    \node[ipe node, anchor=north west]
       at (152, 580) {
         \begin{minipage}{80bp}\kern0pt
           $\Theta(n+m)$
         \end{minipage}
       };
    \node[ipe node]
       at (140, 580) {\Or};
    \node[ipe node]
       at (212, 516) {\Star};
    \node[ipe node]
       at (212, 580) {\Conc};
    \node[ipe node]
       at (212, 628) {\Plus};
    \node[ipe node, anchor=north west]
       at (224, 636) {
         \begin{minipage}{96bp}\kern0pt
           Simplifies \\
           Lem.~\ref{lem:prelim:simplification}
         \end{minipage}
       };
    \node[ipe node, anchor=north west]
       at (224, 528) {
         \begin{minipage}{96bp}\kern0pt
           Complete Subtree \\
           $\Theta(n+m)$ \\
           \cite{BringmannGL17}
         \end{minipage}
       };
    \node[ipe node, anchor=north west]
       at (440, 540) {
         \begin{minipage}{108bp}\kern0pt
           \color{snd} $\Theta\left(\frac{nm}{\poly \log n}\right)$ (\Conc\Or\Conc) \\
           Sec.~\ref{lower:concOrConc}, Lem.~\ref{lem:prelim:hardness}
         \end{minipage}
       };
    \node[ipe node, anchor=north west]
       at (440, 508) {
         \begin{minipage}{108bp}\kern0pt
           \color{snd} $\Theta\left(\frac{nm}{\poly \log n}\right)$ (\Conc\Star) \\
           Sec.~\ref{lower:concStar}, Lem.~\ref{lem:prelim:hardness}
         \end{minipage}
       };
    \node[ipe node, anchor=north west]
       at (440, 572) {
         \begin{minipage}{108bp}\kern0pt
           \color{snd} $\Theta\left(\frac{nm}{\poly \log n}\right)$ (\Conc\Or\Plus) \\
           Sec.~\ref{lower:concOrPlus}, Lem.~\ref{lem:prelim:hardness}
         \end{minipage}
       };
    \draw
      (392, 528)
       -- (436, 528);
    \node[ipe node]
       at (428, 532) {\Conc};
    \node[ipe node]
       at (428, 500) {\Star};
    \node[ipe node]
       at (428, 564) {\Plus};
    \node[ipe node, anchor=north west]
       at (440, 668) {
         \begin{minipage}{108bp}\kern0pt
           \color{snd} $\Theta\left(\frac{nm}{\poly \log n}\right)$ (\Conc\Plus\Conc) \\
           Sec.~\ref{lower:concPlusConc}, Lem.~\ref{lem:prelim:hardness}
         \end{minipage}
       };
    \node[ipe node, anchor=north west]
       at (440, 636) {
         \begin{minipage}{108bp}\kern0pt
           \color{snd} $\Theta\left(\frac{nm}{\poly \log n}\right)$ (\Conc\Star) \\
           Sec.~\ref{lower:concStar}, Lem.~\ref{lem:prelim:hardness}
         \end{minipage}
       };
    \node[ipe node, anchor=north west]
       at (440, 604) {
         \begin{minipage}{108bp}\kern0pt
           \color{snd} $\Theta\left(\frac{nm}{\poly \log n}\right)$ (\Conc\Plus\Or) \\
           Sec.~\ref{lower:concPlusOr}, Lem.~\ref{lem:prelim:hardness}
         \end{minipage}
       };
    \draw
      (392, 624)
       -- (436, 624);
    \node[ipe node]
       at (428, 660) {\Conc};
    \node[ipe node]
       at (428, 628) {\Star};
    \node[ipe node]
       at (428, 596) {\Or};
    \draw
      (148, 656)
       -- (136, 656);
    \node[ipe node]
       at (140, 820) {\Plus};
    \node[ipe node]
       at (140, 660) {\Star};
    \draw
      (148, 816)
       -- (136, 816);
    \draw
      (436, 656)
       -- (424, 656)
       -- (424, 592)
       -- (436, 592);
    \draw
      (436, 560)
       -- (424, 560)
       -- (424, 496)
       -- (436, 496);
    \draw
      (320, 624)
       -- (308, 624)
       -- (308, 528)
       -- (320, 528);
    \draw
      (320, 816)
       -- (308, 816)
       -- (308, 752)
       -- (320, 752);
    \draw
      (320, 720)
       -- (308, 720)
       -- (308, 656)
       -- (320, 656);
    \draw
      (220, 784)
       -- (208, 784)
       -- (208, 688)
       -- (220, 688);
    \draw
      (220, 624)
       -- (208, 624)
       -- (208, 512)
       -- (220, 512);
    \draw
      (136, 576)
       -- (148, 576);
  \end{tikzpicture}
  \caption{The classification of the patterns for pattern matching.
  The red bounds are shown in this paper while the blue ones follow as corollaries.}
  \label{fig:intro:matching}
\end{figure}
Further, the dichotomy shows that the type of a pattern has a larger impact on the hardness than the depth.
The alternative as outer operation of the ``easier'' patterns allows us to split the pattern into independent sub-patterns.
This is crucial for the speed-up since pattern matching for \Conc\Plus\ and \Conc\Or\ is near-linear time solvable \cite{BackursI16,ColeH02}.
Contrary almost all hard pattern types have a concatenation as outer operation which does not allow this decomposition into independent problems.
Further, the length of the matched texts can vary largely.
The pattern $(a \Or aba)(b\Or bca)(a\Or ab)$, for example, can match strings of length 3 to 8.
We exploit both properties in our reductions, especially to encode a boolean OR.

In \cref{sec:prelim} we give a formal definition of homogeneous patterns
and state the problems we start reducing from and the ones we reduce to.
We show the algorithms for the upper bounds in \cref{sec:upper}.
In \cref{sec:lowerMatch} we give the improved lower bounds for pattern matching
while the ones for membership are given in \cref{sec:lowerMemb}.

\section{Preliminaries}\label{sec:prelim}

\subparagraph*{Regular Expressions.}
Recall, that patterns over a finite alphabet $\Sigma$ are build recursively from other patterns using the operations \Or, \Conc, \Plus, and \Star.
We construct the patterns and the language of each pattern (i.e.\ the set of words matched by the pattern) as follows.
Each symbol $\sigma \in \Sigma$ is a pattern representing the language $\lang(\sigma) = \{\sigma\}$.
Let in the following $p_1$ and $p_2$ be two patterns.
For the alternative operation we define $\lang(p_1 \Or p_2) = \lang(p_1) \cup \lang(p_2)$.
For the concatenation we define $\lang(p_1 \Conc p_2) = \{w_1w_2 \mid w_1 \in \lang(p_1) \land w_2 \in \lang(p_2)\}$.
For the Kleene Plus we set $\lang(p_1^+) = \{w \mid \exists k\ge 1: \exists w_1, \dots, w_k \in \lang(p_1): w = w_1 \cdots w_k\}$.
With $\varepsilon$ as the empty word we have $\lang(p_1^\Star) = \lang(p_1^+) \cup \{\varepsilon\}$ for the Kleene Star.

Based on this construction it is easy to see patterns as rooted and node-labeled trees
where each inner node is labeled by an operation and the leaves are labeled by symbols.
We call this tree the \emph{parse tree} of a pattern in the following.
Then each node is connected to the node representing the sub-pattern $p_1$ and also for $p_2$ in the case of the binary operations \Conc\ and \Or.
We define the \emph{size} of a pattern to be the number of inner nodes plus the number of leaves in the parse tree.
We extend the definition of the alternative and the concatenation in the natural way to more than two sub-patterns.
To simplify notation we omit the symbol \Conc\ from the patterns in the following.

We call a pattern \emph{homogeneous} if for each level of the parse tree, all inner nodes are labeled with the same operation.
We define the \emph{type} of a homogeneous pattern $p$ to be the sequence of operations from the root of the parse tree of $p$ to the deepest leaf.
The \emph{depth} of a pattern is the depth of the tree, which is equal to the number of operations in the type.
For example, the pattern $[(abc \Or c) (a \Or dc)c(db \Or c \Or bd)]^+$ is of type \Plus\Conc\Or\Conc\ and has depth 4.

\subparagraph*{Relations between Pattern Types.}
Backurs and Indyk showed in \cite{BackursI16} the first quadratic time lower bound for several homogeneous patterns based on SETH.
This classification was completed by the dichotomy result of Bringmann, Gr{\o}nlund, and Larsen in \cite{BringmannGL17}.
As there are infinitely many homogeneous pattern types, they showed linear-time reductions between different pattern types.
By these reductions lower bounds also transfer to other (more complicated) pattern types
and faster algorithms also give improvements for other (equivalent) patterns.
\begin{lemma}[Lemma~1 and~8 in the full version of \cite{BringmannGL17}]\label{lem:prelim:simplification}
  For any type $T$, applying any of the following rules yields a type $T'$
  such that both are equivalent for pattern matching and membership under linear-time reductions, respectively:
  \begin{itemize}
    \item For pattern matching:
    remove prefix \Plus\ and
    replace prefix \Or\Plus\ by \Or.
    \item For membership:
    replace any substring \Plus\Or\Plus\ by \Plus\Or\ and
    replace prefix $r\Star$ by $r\Plus$ for any $r \in \{\Plus, \Or \}^*$.
    \item For both problems: replace any substring $pp$, for any $p \in \{ \Conc, \Or, \Star, \Plus \}$, by $p$.
  \end{itemize}
  We say that $T$ \emph{simplifies} if one of these rules applies.
  Applying these rules in any order will eventually lead to an unsimplifiable type.
\end{lemma}
\begin{lemma}[Lemma~6 and~9 in the full version of \cite{BringmannGL17}]\label{lem:prelim:hardness}
  For types $T$ and $T'$, there is a linear-time reduction from $T$-pattern~matching/membership to $T'$-pattern~matching/membership if one of the following sufficient conditions holds:
  \begin{itemize}
    \item $T$ is a prefix of $T'$,
    \item we may obtain $T'$ from $T$ by replacing a \Star\ by \Plus\Star,
    \item we may obtain $T'$ from $T$ by inserting a \Or\ at any position,
    \item only for membership: $T$ starts with \Conc\ and we may obtain $T'$ from $T$ by prepending a \Plus\ to $T$.
  \end{itemize}
\end{lemma}
Together with the already known sub-quadratic time algorithms for various pattern types \cite{AhoC75,BackursI16,BringmannGL17,ColeH02,KnuthMP77},
it suffices to check the remaining cases in \cref{tab:intro:hard} to get a fine-grained dichotomy for the hard pattern types
(i.e.\ the ones requiring essentially quadratic time under SETH).

\subparagraph*{Hypothesis.}
As mentioned in the introduction,
we follow the ideas of Abboud and Bringmann in \cite{AbboudB18} and show reductions from \FormSat to pattern matching to prove lower bounds.
Likewise as in their result, we also start from the intermediate problem \FormPair:
Given a \emph{monotone} De~Morgan formula $F$ with size $s$,
that is a De~Morgan formula where each leaf is labeled with a variable, i.e.\ no negation allowed, and each variable is used only once.
Further, one is given two sets $A,B$ of half-assignments to $s/2$ variables of $F$ with $\abs{A}=n$ and $\abs{B}=m$.
The task is to find a pair $a \in A, b \in B$ such that $F(a,b)=\true$.

There is an intuitive reduction from \FormSat to \FormPair as shown in \cite{AbboudB18}.
Thus, FSH implies the following hypothesis, which we prove in \cref{appendix:FSHtoFPH}:
\begin{hypothesis}[\FPH~(FPH)]\label{hypo:FPH}
  For all $k\ge 1$,
  there is no algorithm that can solve \FormPair for a monotone De~Morgan formula $F$ of size $s$
  and sets $A,B \subseteq \SetB^{s/2}$ of size $n$ and $m$, respectively,
  in time $\O({n m s^k} / {\log^{3k+2}n})$
  in the Word-RAM model.
\end{hypothesis}

\subparagraph*{\BOV.}
For the upper bounds we transform texts and patterns into bit-vectors such that they are orthogonal if and only if the text is matched by the pattern.
This gives us a reduction from pattern matching to \OV (OV) (\cite{ChanW16,Williams14}).
But to improve the runtime we process many texts simultaneously using the following lemmas.
\begin{lemma}[\BOV (cf.\ \cite{ChanW16})]\label{lem:batchOV}
  Let $A, B \subseteq \SetB^d$ with $\abs{A}=\abs{B}=n$ and $d\le2^{c^{-1} \sqlog n}$ for some constant $c>0$.
  We can decide for all vectors $a \in A$ whether there is a vector $b \in B$ such that $\langle a,b \rangle=0$ in time $n^2/2^{\epsilon c \sqlog{n}}$ for sufficiently small $\epsilon>0$.
\end{lemma}
We generalise this balanced case to the unbalanced case which we use later:
\begin{lemma}[Unbalanced \BOV]\label{lem:batchOVgeneral}
  Let $A, B \subseteq \SetB^d$ with $\abs{A}=n$ and $\abs{B}=m$ and $d\le2^{c^{-1} \sqlog{\min(n,m)}}$ for some constant $c>0$.
  We can decide for all vectors $a \in A$ whether there is a vector $b \in B$ such that $\langle a,b \rangle=0$ in time $nm/2^{\epsilon c \sqlog{\min(n,m)}}$ for sufficiently small $\epsilon>0$.
\end{lemma}
\begin{proof}
  If $n \le m$,
  partition $B$ into $\lceil m/n\rceil$ sets of size $n$ and run the algorithm from \cref{lem:batchOV} on every instance in time
  $
    \lceil {m}/{n} \rceil {n^2}/{2^{\epsilon c \sqlog{n}}}
    \approx {nm}/{2^{\epsilon c \sqlog{n}}}
  $.
  Analogously for $n > m$.
\end{proof}

\section{Upper Bounds}\label{sec:upper}
For patterns $p$ of type \Or\Conc\Or\ and \Or\Conc\Plus\
let $p=(p_1 \Or p_2 \Or \dots \Or p_k)$ be the pattern of size $m$.
Likewise for the patterns with a Kleene Plus as additional outer operation.
Let further $t=t_1 \cdots t_n$ be the text of length $n$.
The main idea of the fast algorithm is to compute a set of matched substrings:
$ M=\{ (i,j) \mid \exists \ell \in [k]: t_i\cdots t_j \in \lang(p_\ell) \} \subseteq [n]\times[n] $.
From $M$ we construct a graph
where the nodes correspond to different prefixes that can be matched.
The tuples in $M$ represent edges between these nodes.
Then it remains to check whether the node corresponding to $t$ is reachable.
\begin{theorem}[Upper Bounds]\label{thm:upper:main}
  We can solve in time $nm/2^{\Omega \sqlog{\min(n,m)}}$:
  \begin{enumerate}
    \item\label{thm:upper:main:concOr}
    \Or\Conc\Or-pattern matching and
    \Plus\Or\Conc\Or-membership.
    \item\label{thm:upper:main:concPlus}
    \Or\Conc\Plus-pattern matching and
    \Plus\Or\Conc\Plus-membership.
  \end{enumerate}
\end{theorem}
To compute $M$ we split the patterns into large and small ones.
For the large patterns we compute the corresponding values of $M$ sequentially
while for the small patterns we reduce to unbalanced \BOV and use the fast algorithm for this problem shown in \cref{lem:batchOVgeneral}.

\subsection{Patterns of Type \texorpdfstring{\Plus\Or\Conc\Or}{+|o|} and \texorpdfstring{\Or\Conc\Or}{|o|}}
As mentioned in the beginning of this section,
we compute the set $M$ of matched substring by partitioning the sub-patterns into large and small ones.
\begin{lemma}\label{lem:upper:concOr:main}
  Given a text $t$ of length $n$ and patterns $\{p_i\}_i$ of type \Conc\Or\ such that $\sum_i \abs{p_i} = m$.
  We can compute $M$ in time $nm / 2^{\Omega(\sqlog{\min(n,m)})}$.
\end{lemma}
\begin{lemma}[Large Sub-Patterns]\label{lem:upper:concOr:large}
  Given a text $t$ of length $n$
  and patterns $p_1,\dots,p_\ell$ of type \Conc\Or\ such that $\sum_{i=1}^\ell \abs{p_i} \le m$.
  We can compute $M$ in time $\O(\ell n \log^2 \min(n,m) + m)$.
\end{lemma}
\begin{proof}
  From a result by Cole and Hariharan \cite{ColeH02} we know that
  there is a $\O(n \log^2 \hat m + \hat m)$ time algorithm for \Conc\Or-pattern matching with patterns of size $\hat m$.
  We run this algorithm sequentially for every pattern.
  We can ignore all $p_i$ with $\abs{p_i}>\abs{\Sigma}n$ since they match more than $n$ symbols.
  We get $\abs{p_i} \le \min(\abs{\Sigma} n, m) \le \min(n^2, m)\le \min(n^2,m^2)$.
  Since $\log\min(n^2,m^2) = 2\log\min(n,m)$,
  each iteration takes time $\O(n \log^2 \min(n,m) + \abs{p_i})$
  and the claim follows.
\end{proof}
\begin{lemma}[Small Sub-Patterns]\label{lem:upper:concOr:small}
  Given a text $t$ of length $n$ and patterns $p_1, \dots, p_m$ of type \Conc\Or.
  There is a $f \in 2^{\Omega(\sqlog{\min(n,m)})}$ such that the following holds:
  If $\abs{p_i} \le f$ for all $i\in[m]$,
  then we can compute $M$ in time $nm / 2^{\Omega(\sqlog{\min(n,m)})}$
  with small error probability.
\end{lemma}
We postpone the proof of this lemma and first combine the results for small and large patterns to proof the main theorem.
\begin{proof}[Proof of \cref{lem:upper:concOr:main}]
  Choose $f\in 2^{\Omega(\sqlog{\min(n,m)})}$ as in \cref{lem:upper:concOr:small}
  and split the patterns into large patterns of size $>f$ and small patterns of size $\le f$.

  For the at most $m/f$ large patterns
  compute $M_>$ by \cref{lem:upper:concOr:large} in time $\O(m/f \cdot n\log^2\min(n,m) + m)\in nm/2^{\Omega(\sqlog{\min(n,m)})}$.
  Duplicate the $\ell$ small patterns $m/\ell$ times
  and compute $M_\le$ for the $m$ small patterns by \cref{lem:upper:concOr:small} in the claimed running time.
\end{proof}
\begin{proof}[Proof of \cref{thm:upper:main}~\cref{thm:upper:main:concOr}]
  Construct $M$ by \cref{lem:upper:concOr:main}.
  Check for \Or\Conc\Or-pattern matching whether $M=\emptyset$ since any matched substring is sufficient.

  For \Plus\Or\Conc\Or-membership we construct a graph $G$ with nodes $v_0, \dots, v_n$
  where we put an edge from $v_{i-1}$ to $v_j$ if $(i,j) \in M$.
  Then $v_n$ is reachable from $v_0$ iff there is a decomposition of $t$ into substrings which can be matched by the $p_i$s.
  This reachability check can be performed in time $\O(n+\abs{M})$ by a depth-first search starting from $v_0$.
\end{proof}
For the proof of \cref{lem:upper:concOr:small} we proceed as follows.
For the construction of $M$ for small sub-patterns we define some threshold $f$
and check for every substring of $t$ of length at most $f$ whether there is a pattern that matches this substring.
This check is reduced to \BOV by encoding the substrings and patterns as bit-vectors.

For small alphabets with $\abs{\Sigma} < f$ this encoding is rather simple since we can use a one-hot encoding of the alphabet.
But for larger alphabets this does not work as the dimension of the vectors would increase too much
and the fast algorithm for \BOV could not be used anymore.
Therefore, we define a randomised encoding $\chi$ to ensure that the final bit-vectors are not too large.
For simplicity we can assume $\abs{\Sigma} = \Theta(\min(n,m))$ by padding $\Sigma$ with fresh symbols.
The construction in the following lemma is based on the idea of Bloom-Filters \cite{Bloom70}.
\begin{lemma}[Randomised Characteristic Vector]\label{lem:upper:concOr:cV}
  For a finite universe $\Sigma$ and a threshold $f \le 2^{\O(\sqlog{\abs{\Sigma}})}$
  there is a randomised $\chi: \mathcal{P}(\Sigma) \to \SetB^{d}$ with $d\in \O(f \log\abs{\Sigma})$
  such that for all $\sigma \in \Sigma$ and $S \subseteq \Sigma$ with $\abs{S} \le f$ the following holds:
  \begin{itemize}
    \item
    If $\sigma \in S$, then $\chi(\sigma)\deff \chi(\{\sigma\}) \subseteq \chi(S)$, i.e.\ $\forall i \in [d]: \chi(\{\sigma\})[i]=1 \implies \chi(S)[i]=1$.
    \item
    If $\chi(\sigma) \subseteq \chi(S)$, then $\sigma \in S$ with high probability,
    i.e.\ $\ge 1-1/\poly(\abs{\Sigma})$.
  \end{itemize}
\end{lemma}
\begin{proof}
  We define $\chi$ element-wise and
  set for $S \subseteq \Sigma$:
  $\chi(S)[i] \deff \bigvee_{s\in S}\chi(s)[i]$, i.e.\ the bitwise OR over $\chi(s)$ for $s \in S$.
  Hence, the first claim already holds by definition.
  For each $\sigma\in \Sigma$ we define $\chi(\sigma)$ independently by setting $\chi(\sigma)[i]=1$ with probability $1/f$ for all $i \in [d]$.
  Let $S\subseteq \Sigma$ with $\abs{S} \le f$ and $\sigma \in \Sigma\setminus S$.
  For all $i \in [d]$:
  \begin{align*}
    \prob{\chi(\sigma)[i] \nsubseteq \chi(S)[i]}
    &= \prob{\chi(\sigma)[i]=1 \land \chi(S)[i]=0}\\
    &= \tfrac 1 f {\left(1- \tfrac 1 f\right)}^{\abs{S}}
    \ge \tfrac 1 f {\left(1-\tfrac 1 f\right)}^f
    \ge \tfrac {\ee^{-2}} f \\
    \prob{\chi(\sigma) \subseteq \chi(S)}
    &= \prod_{i=1}^d \prob{\chi(\sigma)[i] \subseteq \chi(S)[i]}
    = \prod_{i=1}^d (1-\prob{\chi(\sigma)[i] \nsubseteq \chi(S)[i]}) \\
    & \le \prod_{i=1}^d \left(1- \tfrac{\ee^{-2}} {f}\right)
    = {\left(1- \tfrac{\ee^{-2}} {f}\right)}^d
  \shortintertext{Setting $d = f c \ln \abs{\Sigma}$ for some arbitrary $c>\ee^2$, we get:}
    &= \left(1- \tfrac{\ee^{-2}} {f}\right)^{f\cdot c \ln\abs{\Sigma}}
    \le \ee^{-1/\ee^2 \cdot c \ln\abs{\Sigma}}
    = {\abs \Sigma}^{- c/\ee^2}
    = 1/ \poly {\abs \Sigma}
    \qedhere
  \end{align*}
\end{proof}

\begin{proof}[Proof of \cref{lem:upper:concOr:small}]
  Define $f = 2^{\sqrt{\epsilon}/3 \cdot \sqlog{\min(n,m)}}$ with $\epsilon$ as in \cref{lem:batchOVgeneral}
  and let $a$ be some fresh symbol we add to $\Sigma$.
  Let $\chi: \mathcal{P}(\Sigma) \to \SetB^{f^2}$ be as in \cref{lem:upper:concOr:cV}.
  For simplicity one can think of $\chi$ as the one-hot encoding of alphabet $\Sigma$.

  We define
  $T_j \deff \{t_i \cdots t_{i+j-1} \mid 1 \le i \le n-j+1\}$ and
  $P_j \deff \{p_i \mid \lang(p_i) \subseteq \Sigma^j \}$
  for all $j \in [f]$.
  Then replace all symbols and sub-patterns of type \Or\ by bit-vectors by applying $\chi$.
  Finally, pad every vector in $T_j$ and $P_j$ by $f-j$ repetitions of $\chi(a)$
  and flip all values of $P_j$ bit-wise such that 1s become 0s and vice versa.
  Let $T$ be the set of all $\le nf$ modified texts and $P$ be the set of all $m$ transformed patterns.

  We observe that a text-vector in $T$ is orthogonal to a pattern-vector in $P$ iff the original text was matched by the original pattern.
  Since $f \cdot f^2\le 2^{\sqrt\epsilon \sqlog{\min(n,m)}} \le 2^{\sqrt\epsilon \sqlog{\min(nf,m)}}$,
  we can apply \cref{lem:batchOVgeneral} for $T$ and $P$:
  \[
    \frac {nfm} {2^{(\epsilon/\sqrt\epsilon) \sqlog{\min(nf,m)}}}
    \le \frac {nm} {2^{(\sqrt\epsilon - \sqrt\epsilon/3) \cdot \sqlog{\min(n,m)}}}
    \in \frac {nm} {2^{\Omega(\sqlog{\min(n,m)})}}
    \qedhere
  \]
\end{proof}

\subsection{Patterns of Type \texorpdfstring{\Plus\Or\Conc\Plus}{+|o+} and \texorpdfstring{\Or\Conc\Plus}{|o+}}
First observe that even for small patterns $M$ can be too large to be computed explicitly.
For $t=0^n1^n$ and $p=0^+1^+$ we have $M=[1,n]\times[n+1,2n]$
and thus cannot write down $M$ explicitly in time $o(nm)$.

To get around this problem we first define the \emph{run-length encoding} $r(u)$ of a text $u$ as in \cite{BackursI16}:
We have $r(\varepsilon) = \varepsilon$.
For a non-empty string starting with $\sigma$, let $\ell$ be the largest integer such that the first $\ell$ symbols of $u$ are $\sigma$.
Append the tuple $(\sigma, \ell)$ to the run-length encoding and recurse on $u$ after removing the first $\ell$ symbols.
We use the same approach for patterns of type \Conc\Plus.
But if there occurs a $\sigma^+$ during these $\ell$ positions,
we add $(\sigma, \ge \ell)$ to the encoding, otherwise $(\sigma, =\ell)$.
For example, $r(aaa^+b^+bc) = (a,{\ge3})(b,{\ge 2})(c{=1})$.

The idea is to compute a subset of $M$ which only contains those $(i,j)$
such that there is no distinct $(i',j')$ in the subset with $i'\le i$ and $j'\ge j$ and both substrings of $t$ are matched by the same pattern $p_\ell$.
We augment each tuple with two boolean flags, indicating whether the first and last run of the pattern $p_\ell$ contains a Kleene Plus.
From this set $M'\subseteq \SetB\times[n]\times[n]\times\SetB$ we can fully recover $M$.
For our above example we get $M'=\{(1, n,n+1,1)\}$.

\begin{lemma}\label{lem:upper:concPlus:main}
  Given a text $t$ of length $n$ and patterns $\{p_i\}_i$ of type \Conc\Plus\ such that $\sum_i \abs{p_i} = m$.
  We can compute $M'$ in time $nm / 2^{\Omega(\sqlog{\min(n,m)})}$.
\end{lemma}
\begin{lemma}[Large Sub-Patterns]\label{lem:upper:concPlus:large}
  Given a text $t$ of length $n$
  and patterns $p_1,\dots,p_\ell$ of type \Conc\Plus\ such that $\sum_{i=1}^\ell \abs{p_i} \le m$.
  We can compute $M'$ in time $\O(\ell n \log^2 \min(n,m) + m)$.
\end{lemma}
\begin{proof}
  We modify all patterns such that their first and last run is of the form $(\sigma,{=\ell})$, i.e.\ we remove every Kleene Plus from these two runs.
  There is a $\O(n \log^2 \hat m + \hat m)$ time algorithm for \Conc\Plus-pattern matching with patterns of size $\hat m$ shown in \cite{BackursI16}.
  We run this algorithm sequentially for each altered pattern.
  For every tuple $(i,j)$ the algorithm outputs, we add $(f,i,j,e)$ to $M'$
  where $f$ and $e$ are set to 1 iff the first and last run of the pattern contain a Kleene Plus, respectively.

  We can ignore all $p_i$ with $\abs{p_i}>\abs{\Sigma}n$ because they match more than $n$ symbols.
  Since $\abs{p_i} \le \min(\abs{\Sigma} n, m) \le \min(n^2, m) \le \min(n^2,m^2) = 2\log\min(n,m)$,
  each iteration takes time $\O(n \log^2 \min(n,m) + \abs{p_i})$
  and the claim follows.
\end{proof}

\begin{lemma}[Small Sub-Patterns]\label{lem:upper:concPlus:small}
  For a text $t$ of length $n$ and patterns $p_1, \dots, p_m$ of type \Conc\Plus,
  there is a $f \in 2^{\Omega(\sqlog{\min(n,m)})}$ such that the following holds:
  If $\abs{p_i} \le f$ for all $i\in[m]$,
  then we can compute $M'$ in time $nm / 2^{\Omega(\sqlog{\min(n,m)})}$.
\end{lemma}
We postpone the proof of this lemma and first show the final upper bound
as the proof of \cref{lem:upper:concOr:main} also works for \cref{lem:upper:concPlus:main}.
\begin{proof}[Proof of \cref{thm:upper:main}~\cref{thm:upper:main:concPlus}]
  Use \cref{lem:upper:concPlus:main} to construct $M'$
  and check for \Or\Conc\Plus-pattern matching whether $M'=\emptyset$.

  For \Plus\Or\Conc\Plus-membership we define a graph $G=(V,E)$.
  Instead of having nodes $v_0,\dots,v_n$ as for \Plus\Or\Conc\Or-membership we have for each node $v_i$ three versions,
  $V \deff \{ v_i^0, v_i^1, v_i^2 \mid 0\le i \le n \}$.
  The versions correspond to the different ways a suffix or prefix of a run can be matched.
  For node $v_i^0$ we need that all symbols are explicitly matched by a pattern.
  For $v_i^1$ we need that the suffix of the run containing $t_i$ has to be matched by a pattern starting with $t_{i}^+$.
  For $v_i^2$ we say that the prefix has to be matched by a pattern ending with $t_{i-1}^+$.
  Hence, we add edges for the runs simulating the $\sigma^+$ of a pattern:
  For each run $(\sigma,\ell)$ from position $i$ to $j$ in $t$ with $\ell >1$ we add the edges
  $(v_{k-1}^1, v_{k}^1)$ and $(v_{k}^2, v_{k+1}^2)$ to the graph for $i \le k < j$.
  Further, we add edges $(v_i^2, v_i^0)$ and $(v_i^0, v_i^1)$ to change between the states for all $0\le i \le n$.
  While this construction solely depends on the text,
  we add for each $(f, i, j, e)\in M'$ the edge $(v_{i-1}^f, v_j^{2e})$ to the graph.
  We claim that there is a path from $v_0^0$ to $v_n^0$ if and only if $t \in \lang((p_1 \mid \dots \mid p_k)^+)$.
  We prove this claim in \cref{appendix:correctness}.
  See \cref{fig:upper:concPlus:example} for an example of the construction.

  The time for the construction is linear in the output size.
  The graph has $\Theta(n)$ nodes and $\abs{M'} + \O(n)$ edges.
  As the DFS runs in linear time, the overall runtime follows.
\end{proof}
\begin{figure}
  \centering
  \includegraphics{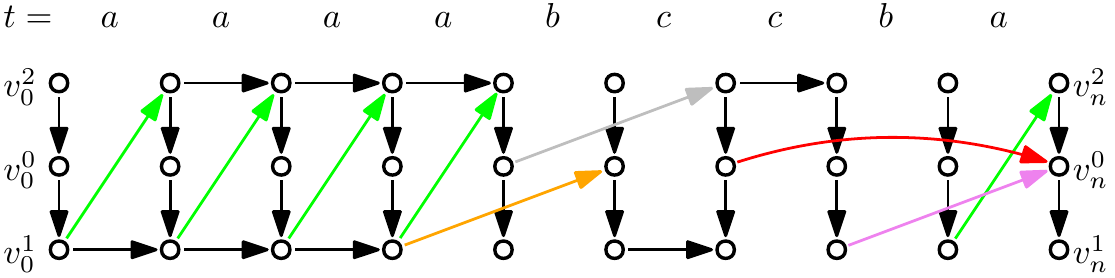}
  \caption{Graph for the pattern
  $({\color{green} a^+} \mid {\color{orange} a^+\!b} \mid {\color{gray}bc^+} \mid {\color{red}cba} \mid {\color{magenta} b^+\!a})^+$ and text $aaaabccba$.}
  \label{fig:upper:concPlus:example}
\end{figure}
It remains to show how the set $M'$ is constructed for small patterns.
\begin{proof}[Proof of \cref{lem:upper:concPlus:small}]
  Set $f \deff 2^{\sqrt\epsilon/5 \sqlog{\min(n,m)}}$ with $\epsilon$ as in \cref{lem:batchOVgeneral}
  and consider all $\le n/f^3$ many long runs of length $\ge f^3$ in $t$.
  Check for each long run by an exhaustive search whether there is a $p_i$ such that the following holds:
  The run in the text is matched by one of the $\le\abs{p_i}$ runs in $p_i$ and the remaining runs of $p_i$ can match the contiguous parts of the text.
  This check can be performed in the following time for all large runs:
  \[
    \frac {n}{f^3} \sum_{i=1}^m \abs{p_i}^2
    \le \frac {n}{f^3} \sum_{i=1}^m f^2
    \le \frac {n m}{f}
  \]
  Since a pattern can have at most $f$ runs and each run matches now at most $f^3$ symbols,
  it remains to check substrings of $t$ of length at most $f^4$.
  Hence, define $T = \{ t_i \cdots t_{i+j-1} \mid \forall j \in [f^4], i \in [n-j+1] \}$
  and ignore all substrings with more than $f$ runs or runs longer than $f^3$.
  Convert these substrings and the patterns into bit-vectors
  by replacing the runs
  by the following bit-vectors of length $2\log \abs{\Sigma} + 2f^3$:
  \begin{align*}
    (c,r) & \mapsto \bin{c} \overline{\bin c} 0^r 1^{f^3-r}1^r 0^{f^3-r}
    &
    (c,{=r}) & \mapsto \overline{\bin c} \bin{c} 1^r 0^{f^3-r}0^r 1^{f^3-r} \\
    &&
    (c,{\ge r}) &\mapsto \overline{\bin c} \bin{c} 1^r 0^{f^3-r}0^{f^3}
  \end{align*}
  $\bin{c}$ denotes the unique binary representation of symbol $c$ and
  $\overline{\bin c}$ its bit-wise negation.
  One can easily see that two such vectors are orthogonal if and only if the runs match each other.
  Thus, a text and a pattern vector resulting from this transformation are orthogonal iff the text is matched by the pattern.
  By padding the vectors with 1s we normalise their length
  but still preserve orthogonality between text and pattern vectors with the same number of runs.
  Let $T'$ and $P'$ be the resulting sets with $\le nf^4$ and $m$ elements, respectively.

  From $\log \abs\Sigma \le \log\min(n,m)\le f$ we get
  $f(2\log \abs\Sigma + 2 f^3) \le f^5 \le 2^{\sqrt\epsilon \sqlog{\min(nf^4,m)}}$
  and hence can apply \cref{lem:batchOVgeneral} for $T'$ and $P'$.
  Actually we have to partition $P'$ depending on whether a pattern has a Kleene Plus in its first and last run.
  Thus, we need four iterations but we can always duplicate patterns such that there are $m$ patterns in each group.
  \[
    \frac{n f^4 m}{2^{\epsilon/\sqrt\epsilon \sqlog{\min(nf^4,m)}}}
    \le
    \frac{n m}{2^{(\sqrt\epsilon-4/5\sqrt\epsilon) \sqlog{\min(n,m)}}}
    \in
    \frac{n m}{2^{\Omega(\sqlog{\min(n,m)})}}
    \qedhere
  \]
\end{proof}

\section{Lower Bounds for Pattern Matching}\label{sec:lowerMatch}
Abboud and Bringmann showed in \cite{AbboudB18} a lower bound for pattern matching (and membership) in general of $\O(nm/\log^{7+\epsilon} n)$,
unless FSH is false.
We use this result and the corresponding reduction as a basis to show similar lower bounds for the remaining hard pattern types.
But we also do not start our reductions directly from \FormSat but from \FormPair as defined in \cref{sec:prelim}
and use the corresponding \FPH from \cref{hypo:FPH}.
\begin{theorem}\label{thm:lower:main}
  There are constants
  $c_\text{\Conc\Star}=76, c_\text{\Conc\Plus\Conc}=c_\text{\Conc\Or\Plus}=72, c_\text{\Conc\Or\Conc}=81$, and $c_\text{\Conc\Plus\Or}=27$
  such that
  pattern matching with patterns of type $T \in \{\text{\Conc\Star}, \text{\Conc\Plus\Conc},\text{\Conc\Or\Conc}, \text{\Conc\Plus\Or}, \text{\Conc\Or\Plus}\}$
  cannot be solved in time $\O({nm}/{\log^{c_T}n})$ even for constant sized alphabets,
  unless FPH is false.
\end{theorem}
We show the lower bounds by a reduction from \FormPair to pattern matching:
\begin{lemma}\label{lem:lower:reduction}
  Given a \FormPair instance with a formula of size $s$, depth $d$, and sets $A$ and $B$ with $n$ and $m\le n$ assignments.
  (If $m > n$, swap $A$ and $B$.)
  We can reduce this to pattern matching with a text $t$ and a pattern $p$ of type
  $T \in \{\text{\Conc\Star}, \text{\Conc\Plus\Conc},\text{\Conc\Or\Conc}, \text{\Conc\Plus\Or}, \text{\Conc\Or\Plus}\}$
  over a constant sized alphabet in time linear in the output size.

  $\abs{t} \in \O(n 5^d s \log s)$ except for \Conc\Plus\Or, there we have $\abs{t}\in \O(n 2^d s \log s)$
  Further, $\abs{p} \in \O(m b_T^d s \log s)$ with
  $b_\text{\Conc\Star}=6, b_\text{\Conc\Plus\Conc}=b_\text{\Conc\Or\Plus}=5, b_\text{\Conc\Or\Conc}=8$, and $b_\text{\Conc\Plus\Or}=1$.
\end{lemma}
\begin{proof}[Proof of \cref{thm:lower:main}]
  We show the result only for patterns of type \Conc\Plus\Conc, the proof for the other types is analogous.

  Let $F$ be a formula of size $s$ with two sets of $n$ half-assignments each, and $d$ be the depth of $F$.
  Applying the depth-reduction technique of Bonet and Buss \cite{BonetB94}
  gives us an equivalent formula $F'$ with size $s' \le s^2$ and depth $d' \le 6 \ln s$.
  By \cref{lem:lower:reduction}
  we get a pattern matching instance with a text $t$ and pattern $p$.
  Both of size
  $\O(n 5^{d'}s'\log s') = \O(n5^{6 \ln s}s^2 \log s) = \O(ns^{6 \ln 5+2}\log s)$.
  Now assume there is an algorithm for pattern matching with the stated running time
  and run it on $t$ and $p$:
  \[ \O\left(\frac{ns^{6 \ln 5+2}\log s \cdot ns^{6 \ln 5+2}\log s}{\log^{72} (ns^{6 \ln 2+2}\log s)} \right)
  \subseteq
  \O\left(\frac{n^2s^{12 \ln 5+4}\log^2 s}{\log^{72} n} \right)
  \subseteq
  \O\left(\frac{n^2s^{23.314}}{\log^{72} n} \right)
  \]
  But this contradicts FPH which was assumed to be true.
\end{proof}

\subsection{Proof of \texorpdfstring{\cref{lem:lower:reduction}}{Lemma \ref{lem:lower:reduction}} for \texorpdfstring{\Conc\Plus\Conc}{o+o}}
\label{lower:concPlusConc}
As the details of the reductions heavily depend on the pattern types, we give each reduction in a separate section.
But we use the reduction for \Conc\Plus\Conc\ as a basis for the other proofs.
For all reductions we first encode the evaluation of a formula on two half-assignments,
then the encoding for finding such a pair.
We define the actual text $t_g$ and the actual pattern $p_g$.
The universal text $u_g$ and universal pattern $q_g$ are needed for technical purposes and do not depend on the assignments.

\subsubsection{Encoding the Formula}
A formula of size $s$ (i.e.\ $s$ leaves) has $s-1$ inner gates and thus $2s-1$ gates in total.
We assign every gate $g$ a unique integer in $[2s-1]$, its ID, and write $\bin{g}$ for the binary encoding of the ID of gate $g$.
We can always see $\bin{g}$ as a sequence of $\lfloor{ \log(2s-1)}\rfloor +1\le\lfloor{\log s}\rfloor+2=\Theta(\log s)$ bits padded with zeros if necessary.
For a fixed gate $g$ we define a separator gadget $G\deff 2\bin{g} 2$ with $2$ as a new symbol.

\begin{description}
  \item [INPUT Gate]
  The text and the pattern depend on the variable that is read:

    For $F_g(a,b)=a_i$
    define $t_g \deff 0a_i1 $ as the text and $p_g \deff 0^+ 1 1^+ $ as the pattern.

    For $F_g(a,b)=b_i$
    define $t_g \deff 0 1 1 $ as the text and $p_g \deff 0^+b_i1^+ $ as the pattern.

  Define $u_g \deff 0011$ as the universal text and $q_g \deff 0^+1^+$ as the universal pattern.

  \item[AND Gates]
  We define:
  $t_g \deff t_1 G t_2 $, $p_g \deff p_1 G p_2 $, $u_g \deff u_1 G u_2 $, and $ q_g \deff q_1 G q_2$.

  \item[OR Gates]
  The texts and the patterns for gate $g$ are defined as follows
  where the parentheses are just for grouping and are not part of the text or pattern:
  \begin{align*}
    t_g\deff && (u_1 GG u_2) G (u_1 GG u_2) G (t_1 G&G t_2) G (u_1 GG u_2) G (u_1 GG u_2) \\
    u_g\deff && (u_1 GG u_2) G (u_1 GG u_2) G (u_1 G&G u_2) G (u_1 GG u_2) G (u_1 GG u_2) \\
    q_g\deff && (u_1 GG u_2) G (u_1 GG u_2) G (q_1 G&G q_2) G (u_1 GG u_2) G (u_1 GG u_2) \\
    p_g\deff && (u_1 GG u_2 G)^+ (q_1 GG p_2) \,\,\,&\!\!\!G (p_1 GG q_2) (G u_1 GG u_2)^+
  \end{align*}
\end{description}

\begin{lemma}[Correctness of the Construction]\label{lem:lower:cpc:form:correct}
  For all assignments $a,b$ and gates $g$:
  \begin{itemize}
    \item $F_g(a, b) = \true \iff t_g(a) \in \lang(p_g(b))$
    \item $t_g(a) \in \lang(q_g)$
    \item $u_g \in \lang(q_g) \cap \lang(p_g(b))$
  \end{itemize}
\end{lemma}
\begin{proof}
  The proofs of the second and third claim follow inductively from the encoding of the gates and especially because of the encoding of the INPUT gate.
  For the first claim we do a structural induction on the output gate of the formula.

  \begin{description}
    \item[INPUT Gate ``$\Rightarrow$'']
    Follows directly from the definition.

    \item[INPUT Gate ``$\Leftarrow$'']
    If the gate is not satisfied,
    then there are not enough 0s or 1s in the text than the pattern has to match.

    \item[AND Gate ``$\Rightarrow$'']
    Follows directly from the definition.

    \item[AND Gate ``$\Leftarrow$'']
    By the uniqueness of the binary encoding, the $G$ in the middle of the text and the pattern have to match.
    Since the whole text is matched, we get $t_1 \in \lang(p_1)$ and $t_2\in \lang(p_2)$
    and $F_g(a,b)$ is satisfied by the induction hypothesis.

    \item[OR Gate ``$\Rightarrow$'']
    $F_g(a,b) = F_{g_1}(a,b) \lor F_{g_2}(a,b) = \true$.
    Assume w.l.o.g.\ that $F_{g_1}(a,b)=\true$, the other case is symmetric.
    Repeat $(u_1 GG u_2 G)^+$ only once to transform $q_1 GG p_2$ into the second $u_1 GG u_2$ by our third claim of the lemma.
    Now $p_1 GG q_2$ matches $t_1 GG t_2$ by the second claim and the assumption $t_1 \in \lang(p_1)$.
    Finally, we match $G u_1 GG u_2 G u_1 GG u_2$ by two repetitions of $(G u_1 GG u_2)^+$.

    \item[OR Gate ``$\Leftarrow$'']
    By the uniqueness of the binary encoding there are exactly 14 $G$s in the text and the pattern can match 11 $G$s when taking both repetitions once.
    Since each additional repetition increases the number by 3, exactly one repetition is taken twice.

    If the first repetition is taken once, the following $q_1 GG p_2$ has to match the second $u_1 GG u_2$ in the text.
    But then $p_1$ is transformed into $t_1$ showing that $F_g$ is satisfied by the inductive hypothesis.
    The case for the second repetition is symmetric.
    \qedhere
  \end{description}
\end{proof}

\subparagraph*{Length of the Text and the Pattern.}
All texts and patterns for a specific gate only depend on the texts and patterns for the two sub-gates.
Thus, we can compute the texts and patterns in a bottom-up manner
and the encoding can be done in time linear in the size of the output.
It remains to analyse the length of the texts and the size of the patterns:
\begin{lemma}\label{lem:lower:cpc:form:size}
  $\abs{u_r},\abs{t_r},\abs{p_r},\abs{q_r} \in \O(5^d s \log s)$.
\end{lemma}
\begin{proof}
  $p_g$ is obviously smaller than $u_g$.
  Since the sizes of $u_g$, $t_g$, and $q_g$ are asymptotically equal,
  it suffices to analyse the length of $u_g$:
  $ \abs{u_g} \le 5 \abs{u_1} + 5 \abs{u_2} + \O(\log s) $.
  Inductively over the $d(F_g)$ levels of $F_g$, i.e.\ the depth of $F_g$, this yields
  $\abs{u_g} \le \O(5^{d(F_g)} s \log s)$.
  The factor of $s \log s$ is due to the $\O(s)$ inner gates each introducing $\O(\log s)$ additional symbols.
\end{proof}

\subsubsection{Final Reduction}
In the first part of the reduction we have seen how to evaluate a formula on one specific pair of half-assignments.
It remains to design a text and a pattern such that such a pair of half-assignments can be chosen.
For this let $A=\{a^{(1)}, \dots, a^{(n)}\}$ be the first set
and $B=\{b^{(1)}, \dots, b^{(m)}\}$ be the second set of half-assignments.
Inspired by the reduction in Section~3.4 in the full version of \cite{BackursI16}
we define the final text and pattern as follows:
\begin{align*}
  t \deff& & &
  \bigodot_{i=1}^{3n} \left(
  33 u_r 3 u_r 3 u_r 3 t(a^{(i)}) 3 u_r 3 u_r 3 u_r 3 u_r
  \right) \\
  p \deff && 3 u_r 3 u_r 3 u_r 3 u_r
  &\bigodot_{j=1}^{m} \left(
  3^+ (u_r 3)^+ u_r 3^+ q_r 3 p(b^{(j)}) 3 (u_r 3)^+ q_r
  \right)
  3 u_r 3 u_r 3 u_r 3 u_r
\end{align*}
Where we set $a^{(j)}=a^{(j \mod n)}$ for $j \in [n+1, 3n]$.
We call the concatenations in $t$ and $p$ for each $i$ and $j$
the $i$th text group and the $j$th pattern group, respectively.

\begin{lemma}\label{lem:lower:cpc:outerOR:completeness}
  If there are $a^{(k)}$ and $b^{(l)}$ such that $F(a^{(k)}, b^{(l)})=\true$, then $t \in \match(p)$.
\end{lemma}
\begin{proof}
  Assume w.l.o.g.\ $a^{(k)}$ and $b^{(k)}$ satisfy $F$.
  Otherwise we have to shift the indices for the text and the pattern accordingly in the proof.
  We match the prefix of $p$ to the suffix of the $n$th text group.
  Then we match the $n+i$th text group by the $i$th pattern group for $i=1, \dots, k-1$:
  Both $(u_r 3)^+$ are repeated twice.
  Then the remaining parts are matched in a straightforward way by transforming the $q_r$s into $t(a^{(i)})$ and $u_r$, and $p(b^{(i)})$ into $u_r$.

  Then, we match the $k$th and $k+1$th pattern group to the $n+k$th text group and a part of the $n+k+1$th text group:
  \begin{center}
    \includegraphics{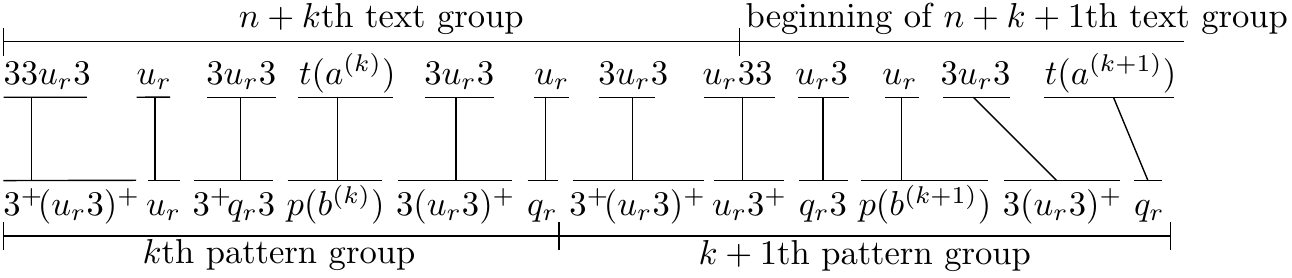}
  \end{center}
  For the last step we shift the groups in the remaining text $t'$ such that it becomes easier to prove which part of the text the remaining pattern matches:
  \begin{align*}
    t' = & 3u_r 3u_r 3u_r 3u_r
    \bigodot_{i=n+k+2}^{3n} \left(33u_r 3u_r 3u_r 3t(a^{(i)}) 3u_r 3u_r 3u_r 3u_r \right) \\
    = & \bigodot_{i=n+k+2}^{3n} \left(3u_r 3u_r 3u_r 3u_r 33u_r 3u_r 3u_r 3t(a^{(i)}) \right)
    3u_r 3u_r 3u_r 3u_r
  \end{align*}
  For each of the remaining pattern groups the first repetition is taken three times.
  With this the $n+i$th group of $t'$ and the $i$th pattern group are matched in a straightforward way for $i=k+2, \dots, m$.
  The suffix of the pattern is matched to the start of the $n+m+1$th text group in the obvious way.
\end{proof}

\begin{lemma}\label{lem:lower:cpc:outerOR:correctness}
  If $t \in \match(p)$, then there are $a^{(k)}$ and $b^{(l)}$ such that $F(a^{(k)}, b^{(l)})=\true$.
\end{lemma}
\begin{proof}
  By the design of the pattern and the text,
  there must be a $j\le n$ such that the prefix of the pattern is matched to the suffix of the $j-1$th text group.
  Likewise the suffix of the pattern has to match the same sequence in some other text group
  because nowhere else the four $3u_r$ could be matched.
  Thus, not all text groups and pattern groups match each other precisely
  and there is a text group $k$ and a pattern group $l$
  such that the pattern group does not match the whole text group or it matches more than this group.
  Choose the first of these groups, i.e.\ the pair with smallest $k$ and $l$.

  Since all prior groups have been matched precisely, the first repetition can be taken at most twice.
  Otherwise the following $u_r$ could not be transformed into a part of the text.
  Now assume it is repeated exactly once.
  Then the following $u_r$ matches the second $u_r$ of the text group.
  Since 3 is a fresh symbol, $q_r$ has to match the third $u_r$.
  But then $p(b^{(k)})$ has to be transformed into $t(a^{(l)})$
  and \cref{lem:lower:cpc:form:correct} gives us a satisfying assignments.

  It remains to check the case when $(u_r3)^+$ is repeated twice.
  Then $q_r$ is transformed into $t(a^{(l)})$ and $p(b^{(k)})$ is transformed into the fourth $u_r$.
  The second repetition has to be taken exactly twice in this case.
  Because otherwise the $33$ from the beginning of the next text group could not be matched.
  But if the pattern $(u_r 3)^+$ is repeated twice, this pattern group is completely matched to a text group, contradicting our assumption.
\end{proof}

\begin{lemma}\label{cor:lower:cpc:outerOR:size}
  The final text has length $\O(n 5^d s \log s)$ and the pattern has size $\O(m 5^d s \log s)$.
\end{lemma}
By this we conclude the proof of \cref{lem:lower:reduction} for this pattern type.
\QED

\subsection{Proof of \texorpdfstring{\cref{lem:lower:reduction}}{Lemma \ref{lem:lower:reduction}} for \texorpdfstring{\Conc\Or\Conc}{o|o}}\label{lower:concOrConc}
When taking a closer look at the reduction for \Conc\Plus\Conc\ one can see that all $\alpha^+$ where only repeated constantly often,
especially at most three times.
Thus, we can replace every $\alpha^+$ by $(\alpha \mid \alpha\alpha \mid \alpha\alpha\alpha)$.
This modification changes the size of the patterns $p_g$
which also dominates the size change for the outer OR.
\begin{lemma}\label{lem:lower:coc:form:size}
  $\abs{u_r},\abs{t_r},\abs{q_r} \in \O(5^d s \log s)$ and $\abs{p_r} \in \O(8^d s \log s)$.
\end{lemma}
\begin{proof}
  Since the size of $p_g$ increased, we get $\abs{u_g},\abs{t_g},\abs{q_g} \in \O(\abs{p_g})$.
  \(
    \abs{p_g}
     \le 6\abs{u_1} + 6\abs{u_2} + \abs{q_1} + \abs{q_2} +\abs{p_1} + \abs{p_2} + \O(\log s)
    \le 8\abs{p_1} + 8\abs{p_2} + \O(\log s)
  \)
  and with the same argument as in the proof of \cref{lem:lower:cpc:form:size}:
  \(
    \abs{p_g}
      \le 8^{d(F_g)} \O(\log s + s \log s)
      = \O(8^{d(F_g)}s \log s)
  \).
  \qedhere
\end{proof}
The correctness follows from the reduction for \Conc\Plus\Conc\
and concludes the proof of \cref{lem:lower:reduction} for \Conc\Or\Conc.
\QED

\subsection{Proof of \texorpdfstring{\cref{lem:lower:reduction}}{Lemma \ref{lem:lower:reduction}} for \texorpdfstring{\Conc\Star}{o*}}\label{lower:concStar}
To reuse the construction from \Conc\Plus\Conc\ for this pattern type,
we first observe that the pattern $\sigma^+$ can be seen as short-hand for $\sigma \sigma^*$.
Hence, the definition of the INPUT and AND gate can be reused.
We also use this idea for the OR gate and simulating $(u_1 GG u_2 G)^+$ by a pattern of type \Conc\Star.
For this we introduce the starred version $\starred{v}$ of a text $v=v_1\dots v_{\abs v}$, where we put a Kleene Star on every symbol:
$\starred{v} \deff v_1^* v_2^* \dots v_{\abs v}^*$.
By this we can reuse $t_g$, $u_g$, and $q_g$, and define for $p_g$:
\[
  p_g\deff (\starred{u_1} \starred{G}\starred{G} \starred{u_2}) \starred{G}
  (u_1 GG u_2) G (q_1 GG p_2) G (p_1 GG q_2) G (u_1 GG u_2)
  \starred{G} (\starred{u_1} \starred{G}\starred{G} \starred{u_2})
\]
\begin{lemma}[Correctness of the Construction]\label{lem:lower:cs:form:correct}
  For all assignments $a,b$ and gates $g$:
  \begin{itemize}
    \item $F_g(a, b) = \true \iff t_g(a) \in \lang(p_g(b))$
    \item $t_g(a) \in \lang(q_g)$
    \item $u_g \in \lang(q_g) \cap \lang(p_g(b))$
  \end{itemize}
\end{lemma}
\begin{proof}
  Again the proof of the last two claims follows directly from the encoding of the gates.
  Since the definition of the INPUT and AND gate is the same as for \Conc\Plus\Conc, we only show the inductive step for the OR gate.
  Recall, that the text is defined as
  \[ t_g\deff (u_1 GG u_2) G (u_1 GG u_2) G (t_1 GG t_2) G (u_1 GG u_2) G (u_1 GG u_2). \]

  \begin{description}
    \item [``$\Rightarrow$'']
    $F_g(a,b) = F_{g_1}(a,b) \lor F_{g_2}(a,b) = \true$.
    Assume w.l.o.g.\ that $F_{g_1}(a,b)=\true$, the other case is symmetric.
    We match the first sequence of starred symbols to the empty string $\epsilon$.
    Then we match $u_1 GG u_2 G$ to each other.
    By the third claim above we can match $u_1 GG u_2$ to $q_1 GG p_2$.
    By the inductive hypothesis and the second claim we match $t_1 GG t_2$ to $p_1 GG q_2$.
    The remaining part of the text is matched in the canonical way to the pattern while the starred sequence matches the original text.

    \item [``$\Leftarrow$'']
    Observe that the $G$s in the pattern have to match $G$s in the text and that the text is matched completely.
    Since the first and last non-starred $GG$ in the pattern have to be matched to a $GG$ in the text,
    one can easily see that both starred sequences either produce the empty string or $u_1 GG u_2 G$ and $G u_1 GG u_2$.
    Thus it remains to check three different cases:
    \begin{itemize}
      \item
      Exactly one sequence produced the empty string.
      Let it w.l.o.g.\ be the first one.
      Then we get that $(u_1 GG u_2) G (t_1 GG t_2)$ has to be matched by $(q_1 GG p_2) G (p_1 GG q_2)$ since $u_1$ and $u_2$ are strings.
      Since the $G$s in the pattern match $G$s in the text, we get $t_1 \in \lang(p_1)$ and thus a satisfying assignment.

      \item
      Both starred sequences produce a non-empty string, i.e.\ their non-starred version.
      The text contains 5 $GG$ but the pattern has to match 6 $GG$.
      A contradiction.

      \item
      Both starred sequences produce the empty string.
      Since $u_1$ and $u_2$ are strings, the remaining text $t'$ has to be matched by the remaining pattern $p'$:
      \begin{align*}
        t' = & (u_1 GG u_2) G (t_1 GG t_2) G (u_1 GG u_2) \\
        p' = & (q_1 GG p_2) G (p_1 GG q_2).
      \end{align*}
      Since the definition of $q_h$ and $u_h$ only differ at the definition of the INPUT gates,
      we we cannot match $q_h$ to something different than $u_h$ here.
      Hence, $u_2 G t_1 GG t_2 G u_1 \in \lang(p_2 G p_1)$.
      Since the number of symbols changes for every word in $\lang(p_2Gp_1)$ is bounded by $A(p_1)+A(p_2)+\ell+2$ with $\ell=A(G)$
      and the text has $2A(u_1)+2A(u_2)+4\ell+7$ symbol changes,
      we get a contradiction by \cref{claim:lower:cs:form:changes}.
      \qedhere
    \end{itemize}
  \end{description}
\end{proof}

\begin{definition}[Symbol Changes]\label{def:symbolChanges}
  We define $A(t)$ to be the number of symbol changes in the text $t$:
  Define $A(\sigma) \deff 0$ for any symbol $\sigma$ and
  $A(t_1\dots t_{n-1}t_n) \deff A(t_1 \dots t_{n-1})+\llbracket t_{n-1} \neq t_n \rrbracket$.
  For patterns $p$ we define $A(p)=\max_{t \in \lang(p)}A(t)$.
\end{definition}
\begin{claim}\label{claim:lower:cs:form:changes}
  $A(u_g)=A(t_g)=A(q_g)$ and $2A(u_g) > A(p_g)$.
\end{claim}
\begin{proof}
  We first observe $A(u_g)=A(t_g)=A(q_g)$ since their definitions only differ for the INPUT gate for which the claim holds.
  We show the main claim by a structural induction on gate $g$.

  For the INPUT gate we have $A(u_g)=A(p_g)=1$ and thus the claim holds.
  For the AND gate the claim follows directly from the induction hypothesis since all texts and patterns start with 0 and end with 1.

  For the OR gate we get $A(u_g)=5A(u_1)+5A(u_2)+14A(G)+18$ and thus:
  \begin{align*}
    2 A(u_g) = & 10A(u_1)+10A(u_2)+28A(G)+36 \\
    = & 5A(u_1)+5A(u_2)+ 5A(u_1)+ 5A(u_2)+28A(G)+36 \\
    \stackrel{\text{IH}}> & 5A(u_1)+5A(u_2)+2.5A(p_1)+ 2.5A(p_2)+28A(G)+36 \\
    > & 5A(u_1)+5A(u_2)+A(p_1)+A(p_2)+17A(G)+22
    = A(p_g)
    \qedhere
  \end{align*}
\end{proof}
With the same arguments as before, we get the following size bounds:
\begin{lemma}\label{lem:lower:cs:form:size}
  $\abs{u_r},\abs{t_r},\abs{q_r} \in \O(5^d s \log s)$ and $\abs{p_r} \in \O(6^d s \log s)$.
\end{lemma}
For the final construction we define a generalised version of the outer OR
that makes use of a helper gadget $H$ that is specific for every type.
 \begin{theorem}\label{thm:lower:outerOR}
  Given $t_r(\cdot)$, $u_r$, $p_r(\cdot)$, and $q_r$ as above.
  Let $H$ be a helper gadget with the following properties:
  \begin{itemize}
    \item $\lang(H) \subseteq \lang(4^+ (3 \Or 4)^* (0 \Or 1 \Or 2 \Or 4)^* (3 \Or 4)^* 4^+)$.
    \item For $\ell\deff \abs{u_r}+4$: $4^{\ell}, 4^{\ell}3u_r 3 4^{\ell} \in \lang(H)$
    \item $\abs{H} \in \O(\abs{u_r})$
  \end{itemize}
  Then we can construct a text $t$ and a pattern $p$
  such that $t \in \match(p)$ if and only if there are $a \in A, b\in B$ such that $F(a,b)=\true$.
  Furthermore, $\abs{t}=\O(n (\abs{u_r}+\abs{t_r}))$, $\abs{p}=\O(m (\abs{u_r}+\abs{p_r}+\abs{q_r}))$
  and $t$ and $p$ are concatenations of gadgets.
\end{theorem}
For \Conc\Star\ we define \( H \deff 44^* 3^* \starred{u_r} 3^* 44^* \).
The proof of the theorem is given in \cref{lower:outerOR}.
\QED

\subsection{Proof of \texorpdfstring{\cref{lem:lower:reduction}}{Lemma \ref{lem:lower:reduction}} for \texorpdfstring{\Conc\Plus\Or}{o+|}}\label{lower:concPlusOr}
Again we only change the encoding of the OR gate and reuse the other parts from \Conc\Plus\Conc.
\begin{align*}
  t_g\deff&& 0 G (t_1 GG u_2) \,\,&\!\!G (u_1 GG t_2) G 1 &&\\
  u_g\deff&& 0 G (u_1 GG u_2) \,\,&\!\!G (u_1 GG u_2) G 1 &&\\
  q_g\deff&& 0 G (q_1 GG u_2) \,\,&\!\!G (u_1 GG q_2) G 1 &&\\
  p_g\deff&& (0 \Or 1 \Or 2)^+ G (p_1 G&G p_2) G (0 \Or 1 \Or 2)^+&&
\end{align*}

\begin{lemma}[Correctness of the construction]\label{lem:lower:cpo:form:correct}
  For all assignments $a,b$ and gates $g$:
  \begin{itemize}
    \item $F_g(a, b) = \true \iff t_g(a) \in \lang(p_g(b))$
    \item $t_g(a) \in \lang(q_g)$
    \item $u_g \in \lang(q_g) \cap \lang(p_g(b))$
  \end{itemize}
\end{lemma}
\begin{proof}
  Again we only show the inductive step for the OR case of the first claim.
  \begin{description}
    \item [``$\Rightarrow$'']
    $F_g(a,b) = F_{g_1}(a,b) \lor F_{g_2}(a,b) = \true$.
    Assume w.l.o.g.\ that $F_{g_1}(a,b)=\true$, the other case is symmetric.
    The first repetition is transformed into the initial 0.
    Then we match $p_1 GG p_2$ to $t_1 GG u_2$ by the third claim and the assumption that $t_1 \in \lang(p_1)$.
    Since the text only consists of symbols from $\{0, 1, 2\}$, the suffix $u_1GGt_2G1$ can be matched by the second repetition.

    \item [``$\Leftarrow$'']
    Since the $GG$ in the pattern has to match one of the two $GG$ in the text,
    there are only two possible ways how the text was matched by the pattern.
    Assume w.l.o.g.\ the $GG$ of the pattern matched the first $GG$ of the text.
    Then the first $G$ of the text and the first $G$ of the pattern match each other.
    Hence, $t_1\in \lang(p_1)$ and
    the induction hypothesis guarantees a satisfying assignment.
    \qedhere
  \end{description}
\end{proof}
\begin{lemma}\label{lem:lower:cpo:form:size}
  $\abs{t_r},\abs{u_r},\abs{q_r} \in \O(2^d s \log s)$ and
  $\abs{p_r} \in \O(s \log s)$.
\end{lemma}
\begin{proof}
  Again we have $\O(\abs{u_g})=\O(\abs{t_g})=\O(\abs{q_g})$.
  For $u_g$ we get:
  \[
    \abs{u_g} \le 2 \abs{u_1} + 2 \abs{u_2} + \O(\log s) \le 2^{d(F_g)} \O(\log s + s \log s) = \O(2^{d(F_g)}s \log s)
  \]
  with the same argument as for the previous size bounds.
  For $p_g$ we have
  $\abs{p_g} \le \abs{p_1}+ \abs{p_2} + \O(\log s) \le \O(s \log s)$.
\end{proof}
We define \( H \deff 4^+ (3 \mid 4)^+ (0 \mid 1 \mid 2 \mid 4)^+ (3 \mid 4)^+ 4^+ \)
and use \cref{thm:lower:outerOR} to conclude the proof of \cref{lem:lower:reduction} for this pattern type.
\QED

\subsection{Proof of \texorpdfstring{\cref{lem:lower:reduction}}{Lemma \ref{lem:lower:reduction}} for \texorpdfstring{\Conc\Or\Plus}{o|+}}\label{lower:concOrPlus}
To reuse the definitions from the previous sections for the last time
we have to allow unary alternatives.
By this we can see a pattern $\sigma^+$ as a pattern of type \Or\Plus.
This is reasonable since we can replace $\sigma^+$ by $(\sigma \Or \sigma^+)$ which represents exactly the same language as just $\sigma^+$.
One could also use a fresh symbol $\alpha$ which will never appear in the text and replace $\sigma^+$ by $(\alpha \Or \sigma^+)$.

We introduce the barred version of a text to match the resulting pattern to the original text but also to the repetition of a single symbol.
\begin{definition}[Barred Version of a Text]\label{def:barred}
  Let $\tau$ be a symbol and $t=t_1 \cdots t_n$ be a text of length $n$.
  Define the barred version of $t$ as a pattern of type \Conc\Or\ as
  \( \overline{t}^\tau \deff(t_1 \Or \tau) \cdots (t_n \Or \tau) \).
\end{definition}
We change the encoding of the OR gate to the following:
\begin{align*}
  t_g\deff && 0^{\abs{u_1 GG u_2 G}+1} (u_1 GG u_2) G (t_1 G&G t_2) G (u_1 GG u_2) 1^{\abs{G u_1 GG u_2}+1} \\
  u_g\deff && 0^{\abs{u_1 GG u_2 G}+1} (u_1 GG u_2) G (u_1 G&G u_2) G (u_1 GG u_2) 1^{\abs{G u_1 GG u_2}+1} \\
  q_g\deff && 0^{\abs{u_1 GG u_2 G}+1} (u_1 GG u_2) G (q_1 G&G q_2) G (u_1 GG u_2) 1^{\abs{G u_1 GG u_2}+1} \\
  p_g\deff && 0^+ \overline{u_1 GG u_2 G}^0 (q_1 GG p_2) \,\,\,&\!\!\!G (p_1 GG q_2) \overline{G u_1 GG u_2}^1 1^+
\end{align*}

\begin{lemma}[Correctness of the construction]\label{lem:lower:cop:form:correct}
  For all assignments $a,b$ and gates $g$:
  \begin{itemize}
    \item $F_g(a, b) = \true \iff t_g(a) \in \lang(p_g(b))$
    \item $t_g(a) \in \lang(q_g)$
    \item $u_g \in \lang(q_g) \cap \lang(p_g(b))$
  \end{itemize}
\end{lemma}
\begin{proof}
  Again we only show the proof for the OR gate in the first claim.
  \begin{description}
    \item [``$\Rightarrow$'']
    $F_g(a,b) = F_{g_1}(a,b) \lor F_{g_2}(a,b) = \true$.
    Assume w.l.o.g.\ that $F_{g_1}(a,b)=\true$, the other case is symmetric.
    We match the first barred text to a repetition of 0s.
    Then $q_1 GG p_2$ matches $u_1 GG u_2$ by the third claim of the lemma.
    $p_1 GG q_2$ matches $t_1 GG t_2$ by the induction hypothesis and the second claim of the lemma.
    The second barred pattern matches its original text while the repetition of 1s is matched by $1^+$.

    \item [``$\Leftarrow$'']
    Since the whole text has to be matched and the $G$s in the pattern have to match $G$s in the text,
    there are three possibilities how the $GG$s of the pattern can be matched to the $GG$s in the text:
    \begin{itemize}
      \item The first $GG$ of the pattern matches the first $GG$ of the text and the second of the text is matched by the second of the pattern.
      This implies $u_2 G t_1 \in \lang(p_2 G p_1)$ and since the $G$ can only match itself,
      $t_1 \in \lang(p_1)$ and a satisfying assignment by the induction hypothesis.

      \item The first $GG$ of the pattern matches the first $GG$ of the text and the second $GG$ of the pattern matches the third $GG$ of the text.
      We get $u_2 G t_1 GG t_2 G u_1 \in \lang(p_2 G p_1)$.
      Using the same argument as for \Conc\Star\
      we get that the number of symbol changes for every word in $\lang(p_2 G p_1)$ is at most $A(p_1)+A(p_2)+A(G)+2$
      while the text has $2A(u_1)+2A(u_2)+4A(G)+6$ symbol changes.
      Analogous to \cref{claim:lower:cs:form:changes} we can show that this case cannot occur since $2 A(u_h) >A(p_h)$.

      \item The first $GG$ of the pattern matches the second $GG$ of the text and the third of the text is matched to the second of the pattern.
      This case is symmetric to the first case and implies $t_2 \in \lang(p_2)$.
      \qedhere
    \end{itemize}
  \end{description}
\end{proof}
\begin{lemma}\label{lem:lower:cop:form:size}
  $\abs{u_r},\abs{t_r},\abs{q_r},\abs{p_r} \in \O(5^d s \log s)$.
\end{lemma}
By defining \( H \deff 4^+ (3 \mid 4) \overline{u_r}^4 (3 \mid 4) 4^+ \)
for the outer OR we finish the proof of \cref{lem:lower:reduction}.
\QED

\subsection{Proof of \texorpdfstring{\cref{thm:lower:outerOR}}{Theorem~\ref{thm:lower:outerOR}}}\label{lower:outerOR}
Let $A=\{a^{(1)}, \dots, a^{(n)}\}$ be the first set
and $B=\{b^{(1)}, \dots, b^{(m)}\}$ be the second set of half-assignments.
Inspired by the reduction in Section~3.6 in the full version of \cite{BackursI16}
we define the final text and pattern as follows:
\begin{align*}
  t \deff&
  \bigodot_{i=1}^{3n} \left(
  333 u_r 3 4^{\ell} 3 u_r 3 4^{\ell}
  33 t_r(a^{(i)}) 3 4^{\ell} 3 u_r 3 4^{\ell}
  \right) \\
  p \deff & 3
  \bigodot_{j=1}^{m} \left(
  33^+ q_r 3 H 3^+ p_r(b^{(j)}) 3 H
  33^+ q_r 3 4^{\ell} 3 u_r 3 4^{\ell}
  \right)
  33 q_r 3 4^{\ell} 3 u_r 3 4^{\ell} 333
\end{align*}
Where $\ell\deff \abs{u_r}+4$ and $a^{(j)}=a^{(j \mod n)}$ for $j \in [n+1, 3n]$.
Again we call the concatenations in $t$ and $p$ for each $i$ and $j$
the $i$th text group and the $j$th pattern group, respectively.
Recall, that we have the following assumption for $H$:
\begin{itemize}
  \item $\lang(H) \subseteq \lang(4^+ (3 \Or 4)^* (0 \Or 1 \Or 2 \Or 4)^* (3 \Or 4)^* 4^+)$.
  \item $4^{\ell}, 4^{\ell}3u_r 3 4^{\ell} \in \lang(H)$
  \item $\abs{H} \in \O(\abs{u_r})$
\end{itemize}

\begin{lemma}\label{lem:lower:outerOR:completeness}
  If there are $a^{(k)}$ and $b^{(l)}$ such that $F(a^{(k)}, b^{(l)})=\true$, then $t \in \match(p)$.
\end{lemma}
\begin{proof}
  Assume w.l.o.g.\ $a^{(k)}$ and $b^{(k)}$ satisfy $F$.
  Otherwise we have to shift the indices for the text and the pattern in the proof accordingly.
  We match the $i$th pattern group to the $i$th text group for $i=1,\dots,k-1$ using the assumptions
  but for the first pattern group the initial repetition is only taken once because of the prefix of $p$.
  The match is performed straightforward by matching $4^{\ell}$ to $H$.
  $p_r(b^{(i)})$ matches $u_r$ and the second $q_r$ matches $t_r(a^{(i)})$.

  Then we match the $k$th pattern group to the $k$th text group and a part of the $k+1$th text group as follows,
  which is again possible by the assumptions:
  \begin{center}
    \includegraphics{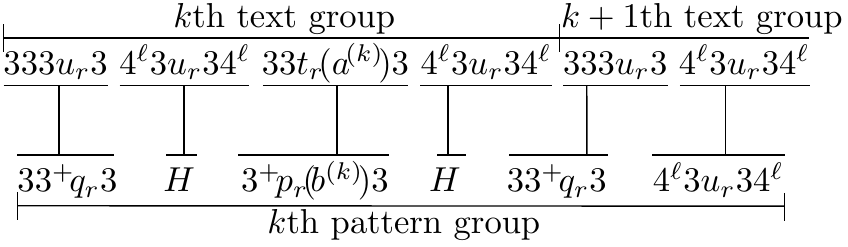}
  \end{center}
  We shift the remaining part $t'$ of the text such that it becomes easier to show which part the remaining pattern matches:
  \begin{align*}
    t' = &
    33 t_r(a^{(k+1)}) 3 4^{\ell} 3 u_r 3 4^{\ell}
    \bigodot_{i=k+2}^{3n} \left(
    333 u_r 3 4^{\ell} 3 u_r 3 4^{\ell}
    33 t_r(a^{(i)}) 3 4^{\ell} 3 u_r 3 4^{\ell}
    \right) \\
    = &
    \bigodot_{i=k+1}^{3n} \left(
    33 t_r(a^{(i)}) 3 4^{\ell} 3 u_r 3 4^{\ell}
    333 u_r 3 4^{\ell} 3 u_r 3 4^{\ell}
    \right)
    33 t_r(a^{(3n)}) 3 4^{\ell} 3 u_r 3 4^{\ell}
  \end{align*}
  After this shift we match the $i$th pattern group to the $i$th text group of $t'$ for $i=k+2, \dots, m$ by
  matching $H$ to $4^{\ell}$ and the other parts in a straightforward way.
  Finally, the suffix of the pattern is matched to the prefix of the $m+1$th text group of $t'$ in the canonical way.
\end{proof}

\begin{lemma}\label{lem:lower:outerOR:correctness}
  If $t \in \match(p)$, then there are $a^{(k)}$ and $b^{(l)}$ such that $F(a^{(k)}, b^{(l)})=\true$.
\end{lemma}
\begin{proof}
  By the design of the pattern and the text,
  there must be a $j\le n$ such that initial ``3'' of the pattern is matched to the first ``3'' of the $j$th text group.
  Furthermore, we know that the suffix of the pattern has to match the suffix of some text group and the following 333.
  Hence, not all pattern groups match exactly one text group but only a prefix or more than one text group.
  We choose the first of these groups (i.e.\ the pair with smallest $k$ and $l$).

  Since all prior groups have been matched precisely,
  $q_r$ has to be transformed into $u_r$ because ``3'' is a fresh symbol.
  Since the 3s have to be aligned, $H$ can only match $4^{\ell}3u_r34^{\ell}$ or just $4^{\ell}$ by assumption.
  In the first case $p_r(b^{(l)})$ matches $t_r(a^{(k)})$ and we get a satisfying assignment.

  Assume for contradictions sake that $H$ matches $4^{\ell}$.
  Then $p_r(b^{(l)})$ matches $u_r$.
  If the second $H$ matches $4^{\ell}$, the text group is matched precisely by the pattern group and we have a contradiction.
  Thus, we can assume $H$ matches $4^{\ell}33t_r(a^{(k)})34^{\ell}$.
  But then the following $33^+$ in the pattern has to match a single ``3''.
  Again a contradiction.
\end{proof}

\begin{lemma}\label{cor:lower:outerOR:or:size}
  $\abs{t} \in \O(n (\abs{u_r} + \abs{t_r}))$ and $\abs{p} \in \O(m (\abs{u_r} + \abs{p_r} + \abs{q_r}))$.
\end{lemma}
This finishes the proof of \cref{thm:lower:outerOR}.
\QED

\section{Lower Bounds for Membership}\label{sec:lowerMemb}
Instead of giving all reductions from scratch,
we reduce pattern matching to membership
and make use of the results in \cref{lem:lower:reduction}.
By this we get the same bounds as for pattern matching given in \cref{thm:lower:main}.
For the remaining pattern type \Or\Plus\Or\Conc\ we give a new reduction from scratch
which is necessary due to the missing concatenation as outer operation.
\subsection{Reducing Pattern Matching to Membership}\label{lowerMemb:matchToMemb}
\begin{lemma}[Reducing Pattern Matching to Membership]\label{lem:lowerMemb:reduction}
  Given a text $t$ and a pattern $p$ with type in $\{$\Conc\Star, \Conc\Plus\Conc, \Conc\Or\Conc, \Conc\Plus\Or, \Conc\Or\Plus$\}$
  over a constant sized alphabet.

  We can construct a text $t'$ and a pattern $p'$ of the same type as $p$ in linear time such that
  $ t \in \match(p) \iff t' \in \lang(p') $.
  Further, $\abs{t'} \in \O(\abs{t})$ and $\abs{p'} \in \O(\abs{t}+\abs{p})$, except for \Conc\Plus\Or, there we even have $\abs{p'} \in \O(\abs{p})$.
\end{lemma}

\begin{proof}[Proof for Patterns of Type \Conc\Star]
We define $t' \deff t$ and $p'\deff \starred{t}p\starred{t}$ where $\starred{t}$ is the starred text as defined in \cref{lower:concStar}.
Then the claim follows directly.
\end{proof}

\begin{proof}[Proof for Patterns of Type \Conc\Plus\Conc]
  Let $\Sigma=\{1, \dots, s\}$ be the alphabet.
  We first encode every symbol
  such that we can simulate a universal pattern (i.e.\ matching any symbol) by some gadget $U$ of type \Conc\Plus.
  Let $f:\Sigma \to \Sigma^{s+1}$ be this encoding with
  $f(x) = 1 \cdots (x-1) x x (x+1) \cdots s$.
  Since we can extend $f$ in the natural way to texts by applying it to every symbol,
  we can also modify patterns of type \Conc\Plus\Conc\ by applying $f$ to every symbol \emph{without} changing the type.
  After applying $f$ we still have
  $ t \in \match(p) \iff f(t) \in \match(f(p)) $.

  For the step from pattern matching to membership we set $U \deff 1^+ 2^+ \cdots s^+$
  and $R \deff 1 2 \cdots s$.
  Obviously $R\in \lang(U)$ and $f(\sigma)\in \lang(U)$ for all $\sigma\in\Sigma$.
  But we also get $R\notin \lang(f(\sigma))$ since $R$ does not contain a repetition of $\sigma$.
  Finally, we define $t' \deff R^{\abs{t}+1}f(t)R^{\abs{t}+1}$ and $p' = R^+U^{\abs{t}}f(p)U^{\abs{t}}R^+$.
  We claim $t \in \match(p) \iff t' \in \lang(p')$.
  \begin{description}
    \item [``$\Rightarrow$'']
    If $t \in \match(p)$, then there is a substring $\hat t$ of $t$ matched by $p$.
    By the above observations, $f(p)$ matches $f(\hat t)$ which is a substring of $f(t)$.
    Then we use $U^{\abs{t}}$ to match the not matched suffix and prefix of $f(t)$ and a part of $R^{\abs{t}+1}$.
    The remaining repetitions of $R$ are matched by the $R^+$ in the beginning and the end.

    \item [``$\Leftarrow$'']
    If $t' \in \lang(p')$, then $f(p)$ has to match some substring of $f(t)$
    because $R$ cannot be matched by the above observation.
    \qedhere
  \end{description}
\end{proof}

\begin{proof}[Proof for Patterns of Type \Conc\Or\Conc]
  Let $L\deff 2^{\lceil \log \abs{t} \rceil} \in \O(\abs t)$.
  For a set $S$ of symbols, we also write $S$ for the pattern representing the alternative of all symbols in $S$.
  Let $a$ be a new symbol:
  \begin{align*}
    t' \deff & a^{3L-1} t a^{3L-1} \\
    p' \deff & \bigodot_{i=0}^{\log L} \left( a^{2^i} \mid a^{2^{i+1}} \right) (\Sigma\cup \{a \})^{L} p ~ (\Sigma\cup \{a \})^{L} \bigodot_{i=0}^{\log L} \left( a^{2^i} \mid a^{2^{i+1}} \right)
  \end{align*}
  This increases the size of the pattern by an additive term of:
  \[ \O(\abs{\Sigma} \cdot L) + \O\left(\sum_{i=0}^{\log L}2^i+2^{i+1}\right) = \O(\abs{\Sigma}L + L) = \O(\abs{\Sigma}L) \]
  \begin{description}
    \item [``$\Rightarrow$'']
    If $t \in \match(p)$, then there is a substring $t_i\cdots t_j$ of $t$ that is matched to $p$.
    Thus we can match $(\Sigma\cup\{a\})^{i-1}p(\Sigma\cup\{a\})^{n-j}$ to $t$.
    The first $L-i+1$ and the last $L-n+j$ repetitions of $\Sigma\cup\{a\}$ are matched to $a$s.
    Hence there remain at least $2L-1$ and at most $3L-1 \le 4L-1$ $a$s as prefix and suffix.
    We match them to $\bigodot_{i=0}^{\log L} ( a^{2^i} \mid a^{2^{i+1}})$ as follows:

    When allowing empty strings in our pattern we can rewrite the concatenation as follows:
    \[
      \bigodot_{i=0}^{\log L} \left( a^{2^i} \mid a^{2^{i+1}} \right)
      \equiv
      \bigodot_{i=0}^{\log L} a^{2^i} \left( \epsilon \mid a^{2^i} \right)
      \equiv
      a^{2L-1} \bigodot_{i=0}^{\log L} \left( \epsilon \mid a^{2^i} \right)
    \]
    Thus, we can ignore the first part of the pattern since it always matches the first and last $2L-1$ repetitions of $a$.
    It remains to show that the concatenation can match $a^z$ for all $z \in [0,2L-1]$.
    But this directly follows from the binary encoding of a number $z \in [0,2L-1]$
    since the $i$th bit contributes $2^i$ to the sum.
    Thus, we choose $\epsilon$ in the pattern above if and only if the $i$th bit is zero.

    \item [``$\Leftarrow$'']
    If $t' \in \lang(p')$, we know that $p$ matched some substring of $t$ since $p$ cannot match $a$s.
    \qedhere
  \end{description}
\end{proof}

\begin{proof}[Proof for Patterns of Type \Conc\Plus\Or]
  Define $t' \deff 1 t 1$ and $p' \deff \Sigma^+ p \Sigma^+$.
  The claim follows directly since a Kleene Plus matches at least one symbol.
\end{proof}

\begin{proof}[Proof for Patterns of Type \Conc\Or\Plus]
  We define $t' \deff 1^{\abs{t}+1} t 1^{\abs{t}+1}$ and $p' \deff 1^+ \Sigma^{\abs{t}} p \Sigma^{\abs{t}} 1^+$.
  \begin{description}
    \item[``$\Rightarrow$'']
    If $t \in \match(p)$, then $p$ matches the corresponding part in $t'$.
    The not matched prefix of $t$ is matched by the sequence of alternatives.
    The remaining 1s in $t'$ are matched by $\Sigma^+$.

    \item[``$\Leftarrow$'']
    If $t' \in \lang(p')$,
    then $p$ has to match some part of $t$ because the prefix and suffix of $p'$ match at least $\abs{t}+1$ symbols.
    \qedhere
  \end{description}
\end{proof}

\subsection{Patterns of Type \texorpdfstring{\Or\Plus\Or\Conc}{|+|o}}\label{lowerMemb:orPlusOrConc}
Even though the remaining hard pattern type \Or\Plus\Or\Conc\ does not have a concatenation as outer operation,
we can still show a similar lower bound as for the other types.
\begin{theorem}\label{thm:lowerMemb:orPlusOrConc}
  \Or\Plus\Or\Conc-membership cannot be solved in time $\O({nm}/{\log^{17}n})$
  even for constant sized alphabets, unless FPH is false.
\end{theorem}
To proof the theorem it suffices to show that \FormPair can be reduced to membership
with a text of length $\O(n s^2 \log s)$ and a pattern of size $\O(ms^3 \log s)$.
Then the claim directly follows from the definition of FPH as for the other types.

\subparagraph*{Idea of the Reduction.}
  As for the other lower bounds, we first encode the evaluation of the formula on two fixed half-assignments.
  We define for each gate $g$ a text $t_g$
  and two dictionaries $D^M_g$ and $D^S_g$ of words.
  The final dictionary for a gate $g$ is defined as $D_g = \bigcup_{g'\in F_g} D^S_{g'} \cup D^M_{g'}$.
  The final pattern is $D_r^+$ where $r$ is the root of $F$.

  $D^M_g$ corresponds to $p_g$ and allows us to match the whole text $t_g$ if the formula is satisfied.
  The texts of the sub-gates are then matched by the corresponding dictionaries.
  But for the OR gate we have to be able to ignore the evaluation of one sub-formula.
  For this we define the set $D^S_g$ which corresponds to $q_g$ and allows us to match the text independently from the assignments.
  As main idea we include the path from the root of the formula to the current gate in the encoding.
  This trace is appended to the text as a prefix and in reverse as suffix.
  The words in $D^M_g$ for OR gates $g$ allow us to jump to a gate in such a trace of exactly one sub-formula.
  Then we use corresponding words from $D^S$ to propagate this jump to the sub-formulas.
  Because the included trace started at the root, we can proceed to the INPUT gates.
  There we add words to accept all evaluations of the gate.
  For the way back up we add the corresponding words in reverse to the dictionaries.

  We make sure that these words are just used at one specific position by
  embedding the encoding of the corresponding gate in the trace.
  Since the gate number can be made unique these words can only be used at one specific position.
  This procedure allows us to write down the words as a set and not as a concatenation as for the other reductions.

\subparagraph*{Encoding the Formula.}
We identify each gate $g$ with its ID, i.e.\ an integer in $[2s]$.
Let $\bin{g}$ be the binary encoding of the gate ID with $\lfloor\log s\rfloor+2=\Theta(\log s)$ bits padded with zeros if necessary.
Further,
let $h_0, h_1, \dots, h_d$ be the path from the root $r=h_0$ of $F$ to the gate $g=h_d$ of depth $d\ge 0$.
To simplify notation we define $h^g_i = 2 \bin{h_i} \bin{g}2$,
i.e.\ the encoding of the gate on the path and the gate where the path ends. \begin{description}
  \item[INPUT Gates]
  We set $D^S_g \deff \{ h^g_{i} \cdots h^g_d 0 h^g_d \cdots h^g_{i},
  h^g_{i} \cdots h^g_d 1 h^g_d \cdots h^g_{i} \mid i \in [d] \}$.

  For $F_g(a,b)=a_i$, we set
  $ t_g \deff h^g_0 \cdots h^g_d a_i h^g_d \cdots h^g_0$
  and $ D^M_g \deff \{ h^g_0 \cdots h^g_d 1 h^g_d \cdots h^g_0 \} $

  For $F_g(a,b)=b_i$, we set
  $ t_g \deff h^g_0 \cdots h^g_d 1 h^g_d \cdots h^g_0 $
  and $ D^M_g \deff \{h^g_0 \cdots h^g_d b_i h^g_d \cdots h^g_0 \} $

  \item[AND Gate]
  We define the text and the corresponding dictionaries as follows:
  \begin{align*}
    t_g \deff & h^g_0 \cdots h^g_d t_1 t_2 h^g_d \cdots h^g_0 \\
    D^M_g\deff & \{h^g_0 \cdots h^g_d, h^g_d \cdots h^g_0\} \\
    D^S_g\deff &\{h^g_i \cdots h^g_d h^{g_1}_0 \cdots h^{g_1}_{i-1},
    h^{g_1}_{i-1} \cdots h^{g_1}_0 h^{g_2}_0 \cdots h^{g_2}_{i-1},
    h^{g_2}_{i-1} \cdots h^{g_2}_0 h^g_d \cdots h^g_i \mid i \in [d] \}
  \end{align*}

  \item[OR Gate]
  We define the text and the additional dictionaries for $g$ as:
  \begin{align*}
    t_g \deff& h^g_0 \cdots h^g_d t_1 h^g_d t_2 h^g_d \cdots h^g_0 \\
    D^M_g\deff& \{
    h^g_0 \cdots h^g_d,
    h^g_d h^{g_2}_0 \cdots h^{g_2}_d,
    h^{g_2}_d \cdots h^{g_2}_0 h^g_d \cdots h^g_0
    \} \\
    \cup& \{
    h^g_0 \cdots h^g_d h^{g_1}_0 \cdots h^{g_1}_d,
    h^{g_1}_d \cdots h^{g_1}_0 h^g_d,
    h^g_d \cdots h^g_0
    \} \\
    D^S_g\deff& \{
    h^g_i \cdots h^g_d h^{g_1}_0 \cdots h^{g_1}_{i-1},
    h^{g_1}_{i-1} \cdots h^{g_1}_0 h^g_d h^{g_2}_0 \cdots h^{g_2}_{i-1},
    h^{g_2}_{i-1} \cdots h^{g_2}_0 h^g_d \cdots h^g_i  \mid i \in [d] \}
  \end{align*}
\end{description}

\begin{lemma}\label{lem:lowerMemb:opoc:form:correctHelper}
  For all assignments $a,b$ and gates $g$:
  \begin{itemize}
    \item $t_g(a) \in \lang(h^g_0 \cdots h^g_{i-1} (D_g(b))^+ h^g_{i-1} \cdots h^g_0)$ for all $i \in [d]$.
    \item $t_g(a) \notin \lang(h^g_0 \cdots h^g_{i-1} (D_g(b))^+ h^g_{j-1} \cdots h^g_0)$
    for all $i\neq j \in [0,d]$,
    where $h^g_0 h^g_{-1}$ and $h^g_{-1} h^g_0$ denote the empty string.
  \end{itemize}
\end{lemma}
\begin{proof}
  The first claim follows by a structural induction on the output gate using only words from $D^S_{g'}$ for the current gate $g'$.
  Likewise we show the second case by a structural induction on the output gate.
  \begin{description}
    \item[INPUT Gate]
    The statement holds by the definition of the dictionary.

    \item[AND Gate]
    Assume the claim is false for $g$.
    We can only match the ``prefix'' $h^g_i \cdots h^g_d$ with the word $h^g_i \cdots h^g_d h^{g_1}_0 \cdots h^{g_1}_{i-1}$.
    And analogously for the ``suffix''.
    The joining part of $t_1t_2$ has to be matched by some $h^{g_1}_{k-1} \cdots h^{g_1}_0 h^{g_2}_0 \cdots h^{g_2}_{k-1}$ for $k\in [0,\dots,d]$
    (possibly the empty string).
    Hence, $t_1 \in \lang(h^{g_1}_0 \cdots h^{g_1}_{i-1} (D_{g_1}(b))^+ h^{g_1}_{k-1} \cdots h^{g_1}_0)$
    and $t_2\in \lang(h^{g_2}_0 \cdots h^{g_2}_{k-1} (D_{g_2}(b))^+ h^{g_2}_{j-1} \cdots h^{g_2}_0)$.
    But from $i\neq j$ it follows that $k\neq i$ or $k \neq j$ and we have a contradiction to the induction hypothesis for $g_1$ or $g_2$.

    \item[OR Gate]
    The ``prefix'' $h^g_i \cdots h^g_d$ has to be matched by $h^g_i \cdots h^g_d h^{g_1}_0 \cdots h^{g_1}_{i-1}$
    and analogously for the ``suffix''.
    If the joining part of $t_1 h^g_d t_2$ was matched by $h^{g_1}_{k-1} \cdots h^{g_1}_0 h^g_d h^{g_2}_0 \cdots h^{g_2}_{k-1}$ for some $k\in[d]$,
    the same proof as for the AND gate applies.
    Otherwise, either $h^{g_1}_d \cdots h^{g_1}_0 h^g_d$ or $h^g_d h^{g_2}_0 \cdots h^{g_2}_d$ was used.
    Let it w.l.o.g.\ be the first one.
    Since $i\in [0,d]$, we have $i \neq d+1$ and hence a contradiction
    to the inductive hypothesis for $g_1$.
    \qedhere
  \end{description}
\end{proof}

\begin{lemma}[Correctness of the Construction]\label{lem:lowerMemb:opoc:form:correct}
  For all assignments $a,b$ and gates $g$:\\
  $F_g(a, b) = \true \iff t_g(a) \in \lang((D_g(b))^+)$.
\end{lemma}
\begin{proof}
  We proof the claim by an induction on the output gate.
  \begin{description}
    \item[INPUT Gate]
    Follows directly from the construction of the text and the dictionary.

    \item[AND Gate ``$\Rightarrow$'']
    We can use $D_1^+$ and $D_2^+$ to match $t_1$ and $t_2$ by the induction hypothesis, respectively.
    The remaining parts are matched by the words in $D^M_g$.

    \item[AND Gate ``$\Leftarrow$'']
    The initial and last $h^g_0$ of the text have to be matched.
    Since the gate $g$ is part of the encoding, we can only use words from $D_g^M$ for this.
    It follows directly that $t_1$ is matched by words from $D_1$ because the initial $h^{g_1}_0$ has to be matched too
    and the words in $D_g^S$ are not eligible for this.
    The same argument shows that $t_2$ is matched by words from $D_2$.
    Hence, the claim follows by the induction hypothesis.

    \item[OR Gate ``$\Rightarrow$'']
    Assume w.l.o.g.\ that $F_{g_1}(a,b)=\true$, the other case is symmetric.
    We match the prefix of $t_g$ in the obvious way by the corresponding word from $D^M_g$.
    By assumption we match $t_1$ with words from $D_1$.
    The prefix $h^{g_2}_0 \dots h^{g_2}_d$ of $t_2$ is matched by the corresponding word in $D^M_g$.
    By the first claim of the previous lemma, we have $t_2 \in \lang(h^{g_2}_0 \dots h^{g_2}_d (D_{g_2})^+ h^{g_2}_d \dots h^{g_2}_0)$
    and the remaining suffix can be matched by the corresponding word from $D^M_g$.

    \item[OR Gate ``$\Leftarrow$'']
    By \cref{lem:lowerMemb:opoc:form:correctHelper} the joining part of $t_1 h^g_d t_2$ has to be matched by either
    $h^g_d h^{g_2}_0 \dots h^{g_2}_d$ or $h^{g_1}_d \dots h^{g_1}_0 h^g_d$.
    Let it w.l.o.g.\ be the first one.
    Then $t_1$ has to be matched by words from $D_{1}$ again by the lemma.
    The inductive hypothesis gives us a satisfying assignment.
    \qedhere
  \end{description}
\end{proof}

\begin{lemma}\label{lem:lower:opoc:form:size}
  We have the following size bounds:
  \begin{itemize}
    \item $\abs{t_r} \in \O(s d \log s) \subseteq \O(s^2 \log s)$
    \item $\abs{D_r} \in \O(s d) \subseteq \O(s^2)$
    \item $\forall x \in D_r: \abs{x} \in \O(d\log s) \subseteq \O(s\log s)$
  \end{itemize}
\end{lemma}
\begin{proof}
  The lemma follows directly from the definitions and the observations that
  $\abs{t_g} \le \abs{t_1}+\abs{t_2}+ \O(d \log s)$,
  $\abs{D_g^M} \in \O(1)$, and
  $\abs{D_g^S} \in \O(d)$.
  \qedhere
\end{proof}

\subparagraph*{Outer OR.}
Let $A=\{a^{(1)}, \dots, a^{(n)}\}$ be the first set
and $B=\{b^{(1)}, \dots, b^{(m)}\}$ be the second set of half-assignments.
Again we encode $A$ by the text and $B$ by the pattern.
For this we observe that the first step of the reduction produced a pattern of type \Plus\Or\Conc.
Thus, we can use the outer alternative to encode the outer OR to select a specific $b^{(j)}$.
To match the whole text, we blow up the text and the pattern and pad each symbol with three new symbols
such that we can distinguish between the following three matching states:
(1) ignore the padding and match a part of the original text to the original pattern,
i.e.\ we evaluate the formula on two half-assignments.
(2) Match an arbitrary prefix, i.e.\ the symbols before the actual match in state (1).
(3) Match some arbitrary suffix, i.e.\ the symbols after the actual match from state (1).
We allow a change between these states only at the end of a text group
and require that we go through all three states if and only if the text can be matched by the pattern.

\begin{definition}[Blow-Up of a Text]\label{def:blowUp}
  Let $t=t_1 \cdots t_n$ be a text of length $n$ and $u$ be some arbitrary string.
  We define $t \!\Uparrow^u \deff u t_1 u t_2 \cdots u t_n$
  and extend it in the natural way to sets.
\end{definition}
Using this we define the final text and pattern as follows:
\begin{align*}
  t \deff & 563
      \bigodot_{i=1}^{n} \left( t(a^{(i)})3 \!\Uparrow^{456} \right)
      45 \\
  p \deff & p_1^+ \mid p_2^+ \mid \dots \mid p_m^+ \\
  p_j\deff& 5604 \mid 5614 \mid 5624 \mid 5634 \mid 563 \mid
      D_r(b^{(j)}) \!\!\Uparrow^{456} \mid 456345 \mid 6045 \mid 6145 \mid 6245 \mid 6345
\end{align*}

\begin{lemma}\label{lem:lowerMemb:opoc:or:completeness}
  If there are $a^{(k)}$ and $b^{(l)}$ such that $F(a^{(k)}, b^{(l)})=\true$, then $t \in \lang(p)$.
\end{lemma}
\begin{proof}
  It suffices to show that we can match $t$ to $p_l^+$.
  The prefix of $t$ and the first $k-1$ text groups are matched by repetitions of $56x4$
  for values $x \in \{0,1,2,3\}$
  while the last three symbols of the $k-1$th group are matched by $563$.
  This is possible by our blow-up with $456$.
  By \cref{lem:lowerMemb:opoc:form:correct} and the definition of the blow-up we get
  $ t(a^{(k)}) \!\!\Uparrow^{456} \in \lang (
  (D_r(b^{(l)}) \!\!\Uparrow^{456} )^+ ) $.
  The following $456345$ is matched by the corresponding pattern
  while the remaining symbols of the text are matched in a straight forward way by repetitions of $6x45$.
\end{proof}

\begin{lemma}\label{lem:lowerMemb:opoc:or:correctness}
  If $t \in \lang(p)$, then there are $a^{(k)}$ and $b^{(l)}$ such that $F(a^{(k)}, b^{(l)})=\true$.
\end{lemma}
\begin{proof}
  By the structure of the pattern we can already fix $l$.
  As there is no way to match the text just with words $56x4$ or $6x45$,
  the word $563$ must have been used at the end of some group to switch to the first state.
  Hence, let the $k$th text group be the first group not matched by words of the form $56x4$.
  Observe that we cannot directly switch to an application of $6x45$ and thus
  get
  $ t(a^{(k)})\!\!\Uparrow^{456} \in \lang (
  (D_r(b^{(l)})\!\!\Uparrow^{456} )^+ ) $.
  Since the blow-up $456$ always matches each other, we can ignore it and get $t(a^{(k)}) \in \lang(D_r(b^{(l)})^+)$
  proving the claim by \cref{lem:lowerMemb:opoc:form:correct}.
\end{proof}

\begin{corollary}\label{cor:lowerMemb:opoc:or:size}
  The final text has length $\O(n s d \log s)\subseteq \O(n s^2 \log s)$ and the pattern has size $\O(m s d^2 \log s)\subseteq \O(m s^3 \log s)$.
\end{corollary}
This finishes the proof of \cref{thm:lowerMemb:orPlusOrConc}.
\QED

\appendix

\newpage

\section{FSH implies FPH}\label{appendix:FSHtoFPH}
We use the following relation between \FormSat and \FormPair to show that FSH implies FPH:
\begin{lemma}[Weak version of Lemma~B.2 in the full version of \cite{AbboudB18}]\label{lem:formulaSatToPair}
  An instance of \FormSat on a De~Morgan formula of size $s$ over $n$ variables
  can be reduced to an instance of \FormPair with a monotone De~Morgan formula of size $k = \O(s)$ and two sets of size $\O(2^{n/2})$ in linear time.
\end{lemma}
\begin{proof}[Proof Idea]
  Let $F$ be the formula for \FormSat on $n$ variables and size $s$.
  We define $F'$ to be the same formula as $F$ but each leaf is labeled with a different variable and we remove the negations from the leaves.

  For all half-assignments $x$ to the first half of variables of $F$ we construct a new half-assignment $a_x$ for $F'$ as follows:
  Let $l$ be a leaf in $F$ with a variable from the first half of inputs
  and let $l'$ be the corresponding variable/leaf in $F'$.
  We set $a_x[l']=\true$ if and only if $l$ evaluates to $\true$ under $x$.
  We construct the set $B$ analogous for the second half of inputs of $F$.
  Since $F$ has $n$ inputs this results in $2^{n/2}$ assignments for $A$ and $B$.
\end{proof}

\begin{lemma}\label{lem:FSHimpliesFPH}
  FSH implies FPH.
\end{lemma}
\begin{proof}
  Assume FSH holds and FPH is false for some fixed $k\ge 1$.
  Let $F$ be a formula for \FormSat on $N$ inputs and size $s=N^{3+1/(4k)} \in N^{3+\Omega(1)}$.
  By \cref{lem:formulaSatToPair} we transform $F$ into a monotone De~Morgan formula $F'$ of size $s'=\O(s)$
  and two sets with $n,m\in \O(2^{N/2})$ assignments.
  We run the algorithm for \FormPair on this instance to contradict FSH:
  \begin{align*}
    \O\left(\frac{n \cdot m \cdot s'^k}{\log^{3k+2}n}\log^{1+o(1)}2^N\right)
    &\subseteq \O\left(\frac{2^{N/2} 2^{N/2} s^k N^{1.25}}{\log^{3k+2}2^{N/2}}\right)
    = \O\left(2^N \frac{N^{3k+0.25+1.25}}{N^{3k+2} (1/2)^{3k+2}}\right) \\
    &= \O\left(2^N \frac{N^{3k+1.5}}{N^{3k+2}}\right)
    = \O\left(\frac{2^N}{N^{0.5}}\right)
  \end{align*}
  See the following paragraph for the additional factor of $N^{1+o(1)}$.
\end{proof}

As Abboud and Bringmann \cite{AbboudB18} we use the \emph{Word-RAM} model as our computational model.
The word size of the machine will be fixed to $\Theta(\log N)$ many bits for input size $N$.
Likewise we assume several operations that can be performed in time $\O(1)$ (e.g.\ AND, OR, NOT, addition, multiplication, \dots).

While this is sufficient for our reductions, we also need that the operations are robust to a change of the word size to state FPH.
As in \cite{AbboudB18} we require that
we can simulate the operations on words of size $\Theta(\log N)$ on a machine with word size $\Theta(\log \log N)$ in time $(\log N)^{1+o(1)}$.

In the above proof the input size increased from $N$ to $n=2^N$.
Hence, we have to simulate the algorithm for \FormPair with word size $\log n=N$ on a machine with word size $\log N$ to get an algorithm for \FormSat.
Thus, the running time slows down by a factor of $(\log n)^{1+o(1)} = N^{1+o(1)}$.

\newpage

\section{Correctness of the Graph Construction for \texorpdfstring{\Plus\Or\Conc\Plus}{+|o+}-Membership}\label{appendix:correctness}
We show the correctness of the graph construction given in the proof of \cref{thm:upper:main}~\cref{thm:upper:main:concPlus}.
\begin{claim}
  If $t \in \lang(p)$,
  then there is a path from $v_0^0$ to $v_n^0$.
\end{claim}
\begin{claimproof}
  Assume $p = (p_1 \Or \dots \Or p_k)^+$.
  Since $t \in \lang(p)$, we can decompose $t$ into
  $t = \tau_1 \cdots \tau_\ell$ such that for all $l \in [\ell]$ $\tau_l \in \lang(p_{k_l})$ for some $k_l \in [k]$.
  Define $\lambda_l = \abs{\tau_1 \cdots \tau_l}$ as the length of the first $l$ parts of $t$ for all $l\in [\ell]$.
  We claim that if $\tau_1 \cdots \tau_l \in \lang(p)$, then there is a path from $v_0^0$ to $v_{\lambda_l}^0$.

  For $l=0$, the claim is vacuously true as $\varepsilon \notin \lang(p)$.
  Now assume the claim holds for arbitrary but fixed $l$.
  We define $i=\lambda_l+1$ and $j = \lambda_{l+1}$ to simplify notation and get $\tau_{l+1} = t_i \cdots t_j$.
  From $\tau_{l+1} \in \lang(p_{k_{l+1}})$ and \cref{lem:upper:concPlus:main}
  we know $(f, i', j', e) \in M'$ for some $i \le i' \le j' \le j$.
  Further, $f, e$ are set to 1 if and only if the first and last run of $p_{k_{l+1}}$ contains a Kleene Plus, respectively.
  Hence, $v_{j'}^{2e}$ is reachable from $v_{i'-1}^f$.
  Now it suffices to show that
  (1) $v_{i'-1}^f$ is reachable from $v_{i-1}^0$
  and (2) $v_{j}^0$ is reachable from $v_{j'}^{2e}$.
  Then the claim follows inductively as $v_{i-1}^0$ is reachable from $v_0^0$.

  We first show (1).
  If $f=0$, we must have $i=i'$ and the claim holds.
  Thus assume $f=1$.
  We know $\tau_{l+1}=t_{i} \cdots t_j \in \lang(p_{k'})$ and $t_{i'} \cdots t_{j'} \in \lang(p_{k'})$ for some $k' \in [k]$.
  As the first run of $p_{k'}$ contains a Kleene Plus, the symbols, $t_i$, $t_{i+1}, \dots, t_{i'}$ are all equal.
  That is, they form a run from $i$ to $i'$.
  By the construction of the graph, there are edges $(v_{i-1}^1, v_{i}^1),  \dots, (v_{i'-2}^1, v_{i'-1}^1)$.
  But there is also the additional edge $(v_{i-1}^0, v_{i-1}^1)$ proving (1).

  By a symmetric argument one can show claim (2).
\end{claimproof}
\begin{claim}
  If there is a path from $v_0^0$ to $v_n^0$,
  then $t \in \lang(p)$.
\end{claim}
\begin{claimproof}
  First observe that it is not possible to reach $v_n^0$ from $v_0^0$ without using edges introduced by tuples in $M'$.
  Now fix some path $P$ from $v_0^0$ to $v_n^0$ and let $P_1, \dots, P_\ell$ be the edges on the path that are introduced by tuples in $M'$.
  Let $P_l=(v_{i'_l-1}^{f_l}, v_{j'_l}^{2e_l})$, i.e.\ $(f_l, i'_l, j'_l, e_l) \in M'$.

  Assume $j'_0=0$ and $i'_{\ell+1}=n+1$ in the following to simplify notation.
  For each tuple there are two indices $i_l$ and $j_l$ such that $j'_{l-1} \le i_l-1 \le i'_l-1$ and $j'_l \le j_l \le i'_{l+1}-1$
  and the path $P$ goes through $v_{i_l-1}^0$ and $v_{j_l}^0$.
  These nodes exist, as every path from $v_{j'_{l-1}}^{2e_{l-1}}$ to $v_{i'_l-1}^{f_l}$ has to go through some node $v_r^0$.
  We have $j_l+1 = i_{l+1}$ for all $l\in [0,\ell]$ with $j_0 = 0$ and $i_{\ell+1}=n+1$
  and hence, $t = t_{i_1} \cdots t_{j_1} t_{i_2} \cdots t_{j_2} \cdots t_{i_\ell} \cdots t_{j_\ell}$.
  Thus, it suffices to show that for every $l \in [\ell]$ there is a $k' \in [k]$ such that
  $t_{i_l} \cdots t_{j_l} \in \lang(p_{k'})$.

  We fix $l$ in the following and omit it as index to simplify notation.
  By the construction of the graph we have $(f, i', j', e) \in M'$ and hence by \cref{lem:upper:concPlus:main}
  $t_{i'} \cdots t_{j'} \in \lang(p_{k'})$ for some $k' \in [k]$.
  We extend this result and claim $t_{i} \cdots t_{j'} \in \lang(p_{k'})$.
  Recall, that there is a path from $v_{i-1}^0$ to $v_{i'-1}^{f}$ in $P$.
  If $f=0$, then $i'=i$ and the claim follows.
  Otherwise, we know that the first run of $p_{k'}$ contains a Kleene Plus for some symbol $\alpha$.
  As no edge resulting from a tuple in $M'$ can be chosen, the edge $(v_{i-1}^0, v_{i-1}^1)$ is contained in the path $P$.
  By the construction of the graph,
  the sequence $t_{i} \cdots t_{i'}$ is contained in some run $\beta^c$.
  But $\alpha = \beta$ and we get $t_{i} \cdots t_{i'-1} t_{i'} \cdots t_{j'} \in \lang(p_{k'})$.

  We can apply the symmetric argument to show that $t_i \cdots t_{j'} t_{j'+1} \cdots t_{j} \in \lang(p_{k'})$
  proving the claim.
\end{claimproof}

\newpage

\section{Graphical Representation of the Results for Membership}\label{appendix:figures}

\begin{figure}[!h]
  \begin{tikzpicture}[ipe import]
    \node[ipe node, anchor=north west]
       at (224, 624) {
         \begin{minipage}{112bp}\kern0pt
           $\Theta(n+m)$ \\
           immediate
         \end{minipage}
       };
    \node[ipe node, anchor=north west]
       at (224, 760) {
         \begin{minipage}{112bp}\kern0pt
           $\Theta(n+m)$ \\
           \cite{BackursI16}
         \end{minipage}
       };
    \node[ipe node, anchor=north west]
       at (224, 700) {
         \begin{minipage}{104bp}\kern0pt
           \color{snd} $\Theta\left(\frac{nm}{\poly \log n}\right)$ (\Conc\Star) \\
           Sec.~\ref{lowerMemb:matchToMemb}, Lem.~\ref{lem:prelim:hardness}
         \end{minipage}
       };
    \draw
      (220, 752)
       -- (208, 752)
       -- (208, 616)
       -- (220, 616);
    \node[ipe node]
       at (212, 692) {$\Star$};
    \node[ipe node]
       at (212, 756) {$\Plus$};
    \node[ipe node]
       at (212, 620) {$\Or$};
    \draw
      (220, 688)
       -- (196, 688);
    \node[ipe node, anchor=north west]
       at (148, 696) {
         \begin{minipage}{96bp}\kern0pt
           $\Theta(n+m)$ \\
           immediate
         \end{minipage}
       };
    \draw
      (144, 688)
       -- (132, 688)
       -- (132, 420)
       -- (144, 420);
    \draw
      (144, 544)
       -- (120, 544);
    \node[ipe node]
       at (124, 548) {$\Or$};
    \node[ipe node]
       at (136, 548) {$\Star$};
    \node[ipe node]
       at (136, 692) {$\Conc$};
    \node[ipe node, anchor=north west]
       at (148, 552) {
         \begin{minipage}{96bp}\kern0pt
           Simplifies \\
           Lem.~\ref{lem:prelim:simplification}
         \end{minipage}
       };
    \node[ipe node, anchor=north west]
       at (300, 628) {
         \begin{minipage}{108bp}\kern0pt
           \color{snd} $\Theta\left(\frac{nm}{\poly \log n}\right)$ (\Conc\Or\Conc) \\
           Sec.~\ref{lowerMemb:matchToMemb}, Lem.~\ref{lem:prelim:hardness}
         \end{minipage}
       };
    \node[ipe node, anchor=north west]
       at (300, 592) {
         \begin{minipage}{108bp}\kern0pt
           \color{snd} $\Theta\left(\frac{nm}{\poly \log n}\right)$ (\Conc\Star) \\
           Sec.~\ref{lowerMemb:matchToMemb}, Lem.~\ref{lem:prelim:hardness}
         \end{minipage}
       };
    \node[ipe node, anchor=north west]
       at (300, 664) {
         \begin{minipage}{108bp}\kern0pt
           \color{snd} $\Theta\left(\frac{nm}{\poly \log n}\right)$ (\Conc\Or\Plus) \\
           Sec.~\ref{lowerMemb:matchToMemb}, Lem.~\ref{lem:prelim:hardness}
         \end{minipage}
       };
    \draw
      (296, 652)
       -- (284, 652)
       -- (284, 580)
       -- (296, 580);
    \node[ipe node]
       at (288, 620) {$\Conc$};
    \node[ipe node]
       at (288, 584) {$\Star$};
    \node[ipe node]
       at (288, 656) {$\Plus$};
    \draw
      (292, 616)
       -- (272, 616);
    \node[ipe node, anchor=north west]
       at (300, 800) {
         \begin{minipage}{108bp}\kern0pt
           \color{snd} $\Theta\left(\frac{nm}{\poly \log n}\right)$ (\Conc\Plus\Conc) \\
           Sec.~\ref{lowerMemb:matchToMemb}, Lem.~\ref{lem:prelim:hardness}
         \end{minipage}
       };
    \node[ipe node, anchor=north west]
       at (300, 764) {
         \begin{minipage}{108bp}\kern0pt
           \color{snd} $\Theta\left(\frac{nm}{\poly \log n}\right)$ (\Conc\Star) \\
           Sec.~\ref{lowerMemb:matchToMemb}, Lem.~\ref{lem:prelim:hardness}
         \end{minipage}
       };
    \node[ipe node, anchor=north west]
       at (300, 728) {
         \begin{minipage}{108bp}\kern0pt
           \color{snd} $\Theta\left(\frac{nm}{\poly \log n}\right)$ (\Conc\Plus\Or) \\
           Sec.~\ref{lowerMemb:matchToMemb}, Lem.~\ref{lem:prelim:hardness}
         \end{minipage}
       };
    \draw
      (296, 788)
       -- (284, 788)
       -- (284, 716)
       -- (296, 716);
    \node[ipe node]
       at (288, 792) {$\Conc$};
    \node[ipe node]
       at (288, 756) {$\Star$};
    \node[ipe node]
       at (288, 720) {$\Or$};
    \draw
      (296, 752)
       -- (268, 752);
    \node[ipe node, anchor=north west]
       at (224, 376) {
         \begin{minipage}{112bp}\kern0pt
           $\Theta(n+m)$ \\
           immediate
         \end{minipage}
       };
    \node[ipe node, anchor=north west]
       at (224, 500) {
         \begin{minipage}{112bp}\kern0pt
           $\Theta(n+m)$ \\
           \cite{BackursI16}
         \end{minipage}
       };
    \node[ipe node, anchor=north west]
       at (224, 428) {
         \begin{minipage}{104bp}\kern0pt
           Simplifies \\
           Lem.~\ref{lem:prelim:simplification}
         \end{minipage}
       };
    \draw
      (220, 492)
       -- (208, 492)
       -- (208, 368)
       -- (220, 368);
    \node[ipe node]
       at (212, 424) {$\Star$};
    \node[ipe node]
       at (212, 496) {$\Conc$};
    \node[ipe node]
       at (212, 372) {$\Or$};
    \draw
      (220, 420)
       -- (196, 420);
    \node[ipe node, anchor=north west]
       at (148, 428) {
         \begin{minipage}{96bp}\kern0pt
           $\Theta(n+m)$ \\
           immediate
         \end{minipage}
       };
    \node[ipe node]
       at (136, 424) {$\Plus$};
    \node[ipe node, anchor=north west]
       at (300, 380) {
         \begin{minipage}{108bp}\kern0pt
           \color{own} $\Theta\left(\frac{nm}{\poly \log n}\right)$ \\
           Thm.~\ref{thm:lowerMemb:orPlusOrConc}
         \end{minipage}
       };
    \node[ipe node, anchor=north west]
       at (300, 344) {
         \begin{minipage}{108bp}\kern0pt
           Simplifies \\
           Lem.~\ref{lem:prelim:simplification}
         \end{minipage}
       };
    \node[ipe node, anchor=north west]
       at (300, 408) {
         \begin{minipage}{108bp}\kern0pt
           Simplifies \\
           Lem.~\ref{lem:prelim:simplification}
         \end{minipage}
       };
    \draw
      (296, 400)
       -- (284, 400)
       -- (284, 336)
       -- (296, 336);
    \node[ipe node]
       at (288, 372) {$\Conc$};
    \node[ipe node]
       at (288, 340) {$\Star$};
    \node[ipe node]
       at (288, 404) {$\Plus$};
    \draw
      (296, 368)
       -- (272, 368);
    \node[ipe node, anchor=north west]
       at (300, 556) {
         \begin{minipage}{108bp}\kern0pt
           $\O(n\log n + m)$ \\
           \cite{BringmannGL17}
         \end{minipage}
       };
    \node[ipe node, anchor=north west]
       at (300, 504) {
         \begin{minipage}{108bp}\kern0pt
           \color{snd} $\Theta\left(\frac{nm}{\poly \log n}\right)$ (\Conc\Star) \\
           Sec.~\ref{lowerMemb:matchToMemb}, Lem.~\ref{lem:prelim:hardness}
         \end{minipage}
       };
    \node[ipe node, anchor=north west]
       at (300, 452) {
         \begin{minipage}{108bp}\kern0pt
           Expected: \\
           $(n+m)^{1+o(1)}$ \\
           \cite{BringmannGL17}
         \end{minipage}
       };
    \draw
      (296, 544)
       -- (284, 544)
       -- (284, 436)
       -- (296, 436);
    \node[ipe node]
       at (288, 548) {$\Plus$};
    \node[ipe node]
       at (288, 496) {$\Star$};
    \node[ipe node]
       at (288, 440) {$\Or$};
    \draw
      (296, 492)
       -- (268, 492);
    \node[ipe node, anchor=north west]
       at (412, 592) {
         \begin{minipage}{100bp}\kern0pt
           \color{snd}$\Theta\left(\frac{nm}{\poly \log n}\right)$ (\Conc\Plus\Conc) \\
           Sec.~\ref{lowerMemb:matchToMemb}, Lem.~\ref{lem:prelim:hardness}
         \end{minipage}
       };
    \node[ipe node, anchor=north west]
       at (412, 556) {
         \begin{minipage}{100bp}\kern0pt
           \color{snd} $\Theta\left(\frac{nm}{\poly \log n}\right)$ (\Conc\Star) \\
           Sec.~\ref{lowerMemb:matchToMemb}, Lem.~\ref{lem:prelim:hardness}
         \end{minipage}
       };
    \node[ipe node, anchor=north west]
       at (412, 520) {
         \begin{minipage}{100bp}\kern0pt
           \color{snd} $\Theta\left(\frac{nm}{\poly \log n}\right)$ (\Conc\Plus\Or) \\
           Sec.~\ref{lowerMemb:matchToMemb}, Lem.~\ref{lem:prelim:hardness}
         \end{minipage}
       };
    \draw
      (408, 580)
       -- (396, 580)
       -- (396, 508)
       -- (408, 508);
    \node[ipe node]
       at (400, 584) {$\Conc$};
    \node[ipe node]
       at (400, 548) {$\Star$};
    \node[ipe node]
       at (400, 512) {$\Or$};
    \node[ipe node, anchor=north west]
       at (412, 448) {
         \begin{minipage}{100bp}\kern0pt
           \color{snd} $\Theta\left(\frac{nm}{\poly \log n}\right)$ (\Conc\Or\Conc) \\
           Sec.~\ref{lowerMemb:matchToMemb}, Lem.~\ref{lem:prelim:hardness}
         \end{minipage}
       };
    \node[ipe node, anchor=north west]
       at (412, 412) {
         \begin{minipage}{100bp}\kern0pt
           \color{snd} $\Theta\left(\frac{nm}{\poly \log n}\right)$ (\Conc\Star) \\
           Sec.~\ref{lowerMemb:matchToMemb}, Lem.~\ref{lem:prelim:hardness}
         \end{minipage}
       };
    \node[ipe node, anchor=north west]
       at (412, 484) {
         \begin{minipage}{100bp}\kern0pt
           \color{snd} $\Theta\left(\frac{nm}{\poly \log n}\right)$ (\Conc\Or\Plus) \\
           Sec.~\ref{lowerMemb:matchToMemb}, Lem.~\ref{lem:prelim:hardness}
         \end{minipage}
       };
    \draw
      (408, 472)
       -- (396, 472)
       -- (396, 400)
       -- (408, 400);
    \node[ipe node]
       at (400, 440) {$\Conc$};
    \node[ipe node]
       at (400, 404) {$\Star$};
    \node[ipe node]
       at (400, 476) {$\Plus$};
    \draw
      (408, 436)
       -- (380, 436)
       -- (372, 436)
       -- (360, 436);
    \draw
      (408, 544)
       -- (368, 544);
  \end{tikzpicture}
  \caption{The classification of the patterns starting with \Or\ for membership.
  The red bounds are shown in this paper while the blue ones follow as corollaries.}
  \label{fig:membershipOr}
\end{figure}
\begin{figure}[!ht]
  \begin{tikzpicture}[ipe import]
    \node[ipe node, anchor=north west]
       at (288, 452) {
         \begin{minipage}{104bp}\kern0pt
           \color{own} $\Theta\left(\frac{nm}{\poly \log n}\right)$ \\
           Sec.~\ref{lowerMemb:matchToMemb}
         \end{minipage}
       };
    \node[ipe node, anchor=north west]
       at (288, 416) {
         \begin{minipage}{104bp}\kern0pt
           \color{snd} $\Theta\left(\frac{nm}{\poly \log n}\right)$ (\Conc\Star) \\
           Sec.~\ref{lowerMemb:matchToMemb}, Lem.~\ref{lem:prelim:hardness}
         \end{minipage}
       };
    \node[ipe node, anchor=north west]
       at (288, 380) {
         \begin{minipage}{104bp}\kern0pt
           \color{own} $\Theta\left(\frac{nm}{\poly \log n}\right)$ \\
           Sec.~\ref{lowerMemb:matchToMemb}
         \end{minipage}
       };
    \draw
      (284, 440)
       -- (272, 440)
       -- (272, 368)
       -- (284, 368);
    \node[ipe node]
       at (276, 444) {\Conc};
    \node[ipe node]
       at (276, 408) {\Star};
    \node[ipe node]
       at (276, 372) {\Or};
    \draw
      (284, 404)
       -- (240, 404);
    \node[ipe node, anchor=north west]
       at (192, 412) {
         \begin{minipage}{80bp}\kern0pt
           $\Theta(n +m)$ \\
           immediate
         \end{minipage}
       };
    \node[ipe node, anchor=north west]
       at (288, 308) {
         \begin{minipage}{104bp}\kern0pt
           \color{own} $\Theta\left(\frac{nm}{\poly \log n}\right)$ \\
           Sec.~\ref{lowerMemb:matchToMemb}
         \end{minipage}
       };
    \node[ipe node, anchor=north west]
       at (288, 272) {
         \begin{minipage}{104bp}\kern0pt
           \color{snd} $\Theta\left(\frac{nm}{\poly \log n}\right)$ (\Conc\Star) \\
           Sec.~\ref{lowerMemb:matchToMemb}, Lem.~\ref{lem:prelim:hardness}
         \end{minipage}
       };
    \node[ipe node, anchor=north west]
       at (288, 344) {
         \begin{minipage}{104bp}\kern0pt
           \color{own} $\Theta\left(\frac{nm}{\poly \log n}\right)$ \\
           Sec.~\ref{lowerMemb:matchToMemb}
         \end{minipage}
       };
    \draw
      (284, 332)
       -- (272, 332)
       -- (272, 260)
       -- (284, 260);
    \node[ipe node]
       at (276, 300) {\Conc};
    \node[ipe node]
       at (276, 264) {\Star};
    \node[ipe node]
       at (276, 336) {\Plus};
    \draw
      (284, 296)
       -- (240, 296);
    \node[ipe node, anchor=north west]
       at (192, 304) {
         \begin{minipage}{80bp}\kern0pt
           $\Theta(n + m)$ \\
           immediate
         \end{minipage}
       };
    \node[ipe node, anchor=north west]
       at (192, 360) {
         \begin{minipage}{80bp}\kern0pt
           \color{own} $\Theta\left(\frac{nm}{\poly \log n}\right)$ \\
           Sec.~\ref{lowerMemb:matchToMemb}
         \end{minipage}
       };
    \draw
      (188, 404)
       -- (176, 404)
       -- (176, 300)
       -- (188, 300);
    \node[ipe node, anchor=north west]
       at (116, 360) {
         \begin{minipage}{80bp}\kern0pt
           $\Theta(n+m)$ \\
           immediate
         \end{minipage}
       };
    \draw
      (188, 352)
       -- (164, 352);
    \node[ipe node]
       at (96, 356) {\Conc};
    \node[ipe node]
       at (180, 356) {\Star};
    \node[ipe node]
       at (180, 412) {\Plus};
    \node[ipe node]
       at (180, 304) {\Or};
    \draw
      (112, 352)
       -- (88, 352);
    \node[ipe node, anchor=north west]
       at (116, 448) {
         \begin{minipage}{80bp}\kern0pt
           Simplifies \\
           Lem.~\ref{lem:prelim:simplification}
         \end{minipage}
       };
    \draw
      (112, 440)
       -- (88, 440);
    \node[ipe node]
       at (96, 444) {\Star};
    \draw
      (112, 532)
       -- (100, 532)
       -- (100, 724)
       -- (112, 724);
    \node[ipe node]
       at (88, 624) {$\Plus$};
    \node[ipe node]
       at (104, 624) {$\Star$};
    \node[ipe node, anchor=north west]
       at (116, 628) {
         \begin{minipage}{96bp}\kern0pt
           Simplifies \\
           Lem.~\ref{lem:prelim:simplification}
         \end{minipage}
       };
    \node[ipe node, anchor=north west]
       at (192, 580) {
         \begin{minipage}{112bp}\kern0pt
           Simplifies \\
           Lem.~\ref{lem:prelim:simplification}
         \end{minipage}
       };
    \node[ipe node, anchor=north west]
       at (192, 548) {
         \begin{minipage}{112bp}\kern0pt
           Word Break\\
           $\Theta(nm^{1/3}+m)$ \\
           \cite{BackursI16}
         \end{minipage}
       };
    \node[ipe node, anchor=north west]
       at (192, 500) {
         \begin{minipage}{104bp}\kern0pt
           Simplifies \\
           Lem.~\ref{lem:prelim:simplification}
         \end{minipage}
       };
    \draw
      (188, 572)
       -- (176, 572)
       -- (176, 492)
       -- (188, 492);
    \node[ipe node]
       at (180, 496) {$\Star$};
    \node[ipe node]
       at (180, 536) {$\Conc$};
    \node[ipe node]
       at (180, 576) {$\Plus$};
    \draw
      (188, 532)
       -- (160, 532);
    \node[ipe node, anchor=north west]
       at (116, 540) {
         \begin{minipage}{96bp}\kern0pt
           $\Theta(n+m)$ \\
           immediate
         \end{minipage}
       };
    \node[ipe node]
       at (104, 536) {$\Or$};
    \node[ipe node, anchor=north west]
       at (288, 544) {
         \begin{minipage}{108bp}\kern0pt
           \color{snd} $\Theta\left(\frac{nm}{\poly \log n}\right)$ (\Conc\Star) \\
           Sec.~\ref{lowerMemb:matchToMemb}, Lem.~\ref{lem:prelim:hardness}
         \end{minipage}
       };
    \node[ipe node, anchor=north west]
       at (288, 596) {
         \begin{minipage}{108bp}\kern0pt
           \color{own} $\frac{nm}{2^{\Omega(\sqlog{\min(n,m)})}}$ \\
           Thm.~\ref{thm:upper:main}
         \end{minipage}
       };
    \node[ipe node, anchor=north west]
       at (288, 488) {
         \begin{minipage}{108bp}\kern0pt
           \color{own} $\frac{nm}{2^{\Omega(\sqlog{\min(n,m)})}}$ \\
           Thm.~\ref{thm:upper:main}
         \end{minipage}
       };
    \draw
      (284, 584)
       -- (272, 584)
       -- (272, 476)
       -- (284, 476);
    \node[ipe node]
       at (276, 588) {$\Plus$};
    \node[ipe node]
       at (276, 536) {$\Star$};
    \node[ipe node]
       at (276, 480) {$\Or$};
    \draw
      (284, 532)
       -- (260, 532);
    \node[ipe node, anchor=north west]
       at (192, 664) {
         \begin{minipage}{112bp}\kern0pt
           $\Theta(n+m)$ \\
           immediate
         \end{minipage}
       };
    \node[ipe node, anchor=north west]
       at (192, 784) {
         \begin{minipage}{112bp}\kern0pt
           $\Theta(n+m)$ \\
           \cite{BackursI16}
         \end{minipage}
       };
    \node[ipe node, anchor=north west]
       at (192, 736) {
         \begin{minipage}{104bp}\kern0pt
           \color{snd} $\Theta\left(\frac{nm}{\poly \log n}\right)$ (\Conc\Star) \\
           Sec.~\ref{lowerMemb:matchToMemb}, Lem.~\ref{lem:prelim:hardness}
         \end{minipage}
       };
    \draw
      (188, 776)
       -- (176, 776)
       -- (176, 656)
       -- (188, 656);
    \node[ipe node]
       at (180, 728) {$\Star$};
    \node[ipe node]
       at (180, 780) {$\Plus$};
    \node[ipe node]
       at (180, 660) {$\Or$};
    \draw
      (188, 724)
       -- (164, 724);
    \node[ipe node, anchor=north west]
       at (116, 732) {
         \begin{minipage}{96bp}\kern0pt
           $\Theta(n+m)$ \\
           immediate
         \end{minipage}
       };
    \node[ipe node]
       at (104, 728) {$\Conc$};
    \node[ipe node, anchor=north west]
       at (288, 668) {
         \begin{minipage}{108bp}\kern0pt
           \color{snd} $\Theta\left(\frac{nm}{\poly \log n}\right)$ (\Conc\Or\Conc) \\
           Sec.~\ref{lowerMemb:matchToMemb}, Lem.~\ref{lem:prelim:hardness}
         \end{minipage}
       };
    \node[ipe node, anchor=north west]
       at (288, 632) {
         \begin{minipage}{108bp}\kern0pt
           \color{snd} $\Theta\left(\frac{nm}{\poly \log n}\right)$ (\Conc\Star) \\
           Sec.~\ref{lowerMemb:matchToMemb}, Lem.~\ref{lem:prelim:hardness}
         \end{minipage}
       };
    \node[ipe node, anchor=north west]
       at (288, 704) {
         \begin{minipage}{108bp}\kern0pt
           \color{snd} $\Theta\left(\frac{nm}{\poly \log n}\right)$ (\Conc\Or\Plus) \\
           Sec.~\ref{lowerMemb:matchToMemb}, Lem.~\ref{lem:prelim:hardness}
         \end{minipage}
       };
    \draw
      (284, 692)
       -- (272, 692)
       -- (272, 620)
       -- (284, 620);
    \node[ipe node]
       at (276, 660) {$\Conc$};
    \node[ipe node]
       at (276, 624) {$\Star$};
    \node[ipe node]
       at (276, 696) {$\Plus$};
    \draw
      (284, 656)
       -- (240, 656);
    \node[ipe node, anchor=north west]
       at (288, 824) {
         \begin{minipage}{108bp}\kern0pt
           \color{snd} $\Theta\left(\frac{nm}{\poly \log n}\right)$ (\Conc\Plus\Conc) \\
           Sec.~\ref{lowerMemb:matchToMemb}, Lem.~\ref{lem:prelim:hardness}
         \end{minipage}
       };
    \node[ipe node, anchor=north west]
       at (288, 788) {
         \begin{minipage}{108bp}\kern0pt
           \color{snd} $\Theta\left(\frac{nm}{\poly \log n}\right)$ (\Conc\Star) \\
           Sec.~\ref{lowerMemb:matchToMemb}, Lem.~\ref{lem:prelim:hardness}
         \end{minipage}
       };
    \node[ipe node, anchor=north west]
       at (288, 752) {
         \begin{minipage}{108bp}\kern0pt
           \color{snd} $\Theta\left(\frac{nm}{\poly \log n}\right)$ (\Conc\Plus\Or) \\
           Sec.~\ref{lowerMemb:matchToMemb}, Lem.~\ref{lem:prelim:hardness}
         \end{minipage}
       };
    \draw
      (284, 812)
       -- (272, 812)
       -- (272, 740)
       -- (284, 740);
    \node[ipe node]
       at (276, 816) {$\Conc$};
    \node[ipe node]
       at (276, 780) {$\Star$};
    \node[ipe node]
       at (276, 744) {$\Or$};
    \draw
      (284, 776)
       -- (236, 776);
    \node[ipe node, anchor=north west]
       at (396, 632) {
         \begin{minipage}{108bp}\kern0pt
           \color{snd}$\Theta\left(\frac{nm}{\poly \log n}\right)$ (\Conc\Plus\Conc) \\
           Sec.~\ref{lowerMemb:matchToMemb}, Lem.~\ref{lem:prelim:hardness}
         \end{minipage}
       };
    \node[ipe node, anchor=north west]
       at (396, 596) {
         \begin{minipage}{108bp}\kern0pt
           \color{snd} $\Theta\left(\frac{nm}{\poly \log n}\right)$ (\Conc\Star) \\
           Sec.~\ref{lowerMemb:matchToMemb}, Lem.~\ref{lem:prelim:hardness}
         \end{minipage}
       };
    \node[ipe node, anchor=north west]
       at (396, 560) {
         \begin{minipage}{108bp}\kern0pt
           \color{snd} $\Theta\left(\frac{nm}{\poly \log n}\right)$ (\Conc\Plus\Or) \\
           Sec.~\ref{lowerMemb:matchToMemb}, Lem.~\ref{lem:prelim:hardness}
         \end{minipage}
       };
    \draw
      (392, 620)
       -- (380, 620)
       -- (380, 548)
       -- (392, 548);
    \node[ipe node]
       at (384, 624) {$\Conc$};
    \node[ipe node]
       at (384, 588) {$\Star$};
    \node[ipe node]
       at (384, 552) {$\Or$};
    \node[ipe node, anchor=north west]
       at (396, 488) {
         \begin{minipage}{108bp}\kern0pt
           \color{snd} $\Theta\left(\frac{nm}{\poly \log n}\right)$ (\Conc\Or\Conc) \\
           Sec.~\ref{lowerMemb:matchToMemb}, Lem.~\ref{lem:prelim:hardness}
         \end{minipage}
       };
    \node[ipe node, anchor=north west]
       at (396, 452) {
         \begin{minipage}{108bp}\kern0pt
           \color{snd} $\Theta\left(\frac{nm}{\poly \log n}\right)$ (\Conc\Star) \\
           Sec.~\ref{lowerMemb:matchToMemb}, Lem.~\ref{lem:prelim:hardness}
         \end{minipage}
       };
    \node[ipe node, anchor=north west]
       at (396, 524) {
         \begin{minipage}{108bp}\kern0pt
           \color{snd} $\Theta\left(\frac{nm}{\poly \log n}\right)$ (\Conc\Or\Plus) \\
           Sec.~\ref{lowerMemb:matchToMemb}, Lem.~\ref{lem:prelim:hardness}
         \end{minipage}
       };
    \draw
      (392, 512)
       -- (380, 512)
       -- (380, 440)
       -- (392, 440);
    \node[ipe node]
       at (384, 480) {$\Conc$};
    \node[ipe node]
       at (384, 444) {$\Star$};
    \node[ipe node]
       at (384, 516) {$\Plus$};
    \draw
      (392, 476)
       -- (356, 476);
    \draw
      (392, 584)
       -- (356, 584);
    \draw
      (112, 620)
       -- (88, 620);
  \end{tikzpicture}
  \caption{The classification of the patterns starting with \Plus, \Star, or \Conc\ for membership.
  The red bounds are shown in this paper while the blue ones follow as corollaries.}
  \label{fig:membershipPlus}
  \label{fig:membershipStar}
  \label{fig:membershipConc}
\end{figure}
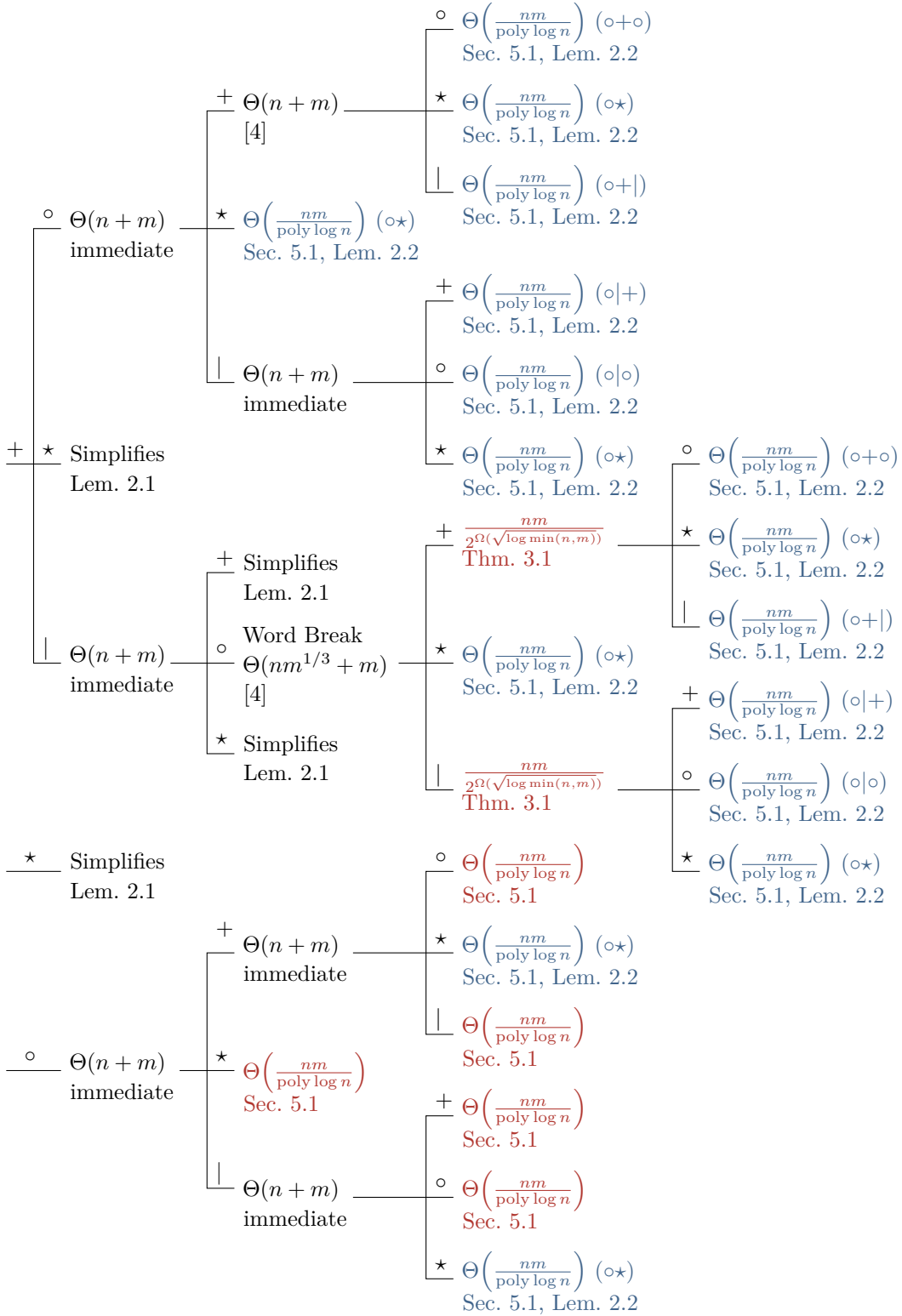

\end{document}